\documentclass[letterpaper,11pt]{article}
\usepackage[utf8]{inputenc}

\pdfoutput=1

\usepackage{amsmath, amsthm, amssymb}
\usepackage{mathtools}
\usepackage[numbers]{natbib}
\usepackage{comment} 
\usepackage{color-edits} 
\usepackage[most]{tcolorbox}
\usepackage{xfrac}
\usepackage[pagebackref=true]{hyperref}
\usepackage{multirow}
\usepackage{caption}
\usepackage{bm}
\usepackage{newfloat}
\usepackage{enumitem}
\usepackage[linesnumbered,ruled,vlined,noend]{algorithm2e}
\usepackage{setspace}
\usepackage[margin=1in]{geometry}

\renewcommand{\qed}{\nobreak \ifvmode \relax \else
      \ifdim\lastskip<1.5em \hskip-\lastskip
      \hskip1.5em plus0em minus0.5em \fi \nobreak
      \vrule height0.75em width0.5em depth0.25em\fi}

\AtBeginEnvironment{algorithm}{
\SetAlCapHSkip{0em}
\DontPrintSemicolon
\setstretch{1.1}
}

\newtheoremstyle{slplain}
  {5pt}
  {3pt}
  {\itshape}
  {}
  {\bfseries}
  {.}
  { }
  {}
\theoremstyle{slplain}

\makeatletter
\renewenvironment{proof}[1][\proofname]{\par
  \vspace{1pt}
  \pushQED{\qed}%
  \normalfont
  \topsep0pt \partopsep0pt 
  \trivlist
  \item[\hskip\labelsep
        \itshape
    #1\@addpunct{.}]\ignorespaces
}{%
  \popQED\endtrivlist\@endpefalse
  \addvspace{6pt} 
}
\makeatother

\allowdisplaybreaks

\title{Time-Optimal Sublinear Algorithms for\\Matching and Vertex Cover\footnote{A preliminary version of this paper appeared in proceedings of FOCS 2021.\vspace{0.2cm}}}

\author{Soheil Behnezhad\footnote{Research was done in part while the author was funded by a Google PhD Fellowship at the University of Maryland.}\\ {\em Stanford University}}

\date{}

\DeclareMathOperator{\polylog}{polylog}

\DeclareMathOperator*{\Exp}{{\normalfont \textbf{E}}}
\DeclareMathOperator*{\Prob}{{\normalfont \textbf{Pr}}}

\newcommand{\E}[0]{\Exp}
\renewcommand{\Pr}[0]{\Prob}

\renewcommand{\epsilon}{\varepsilon}

\newcommand{\GMM}[1]{\ensuremath{\mathsf{GMM}(#1)}}

\newcommand{\cev}[1]{\reflectbox{\ensuremath{\vec{\reflectbox{\ensuremath{#1}}}}}}

\newcommand{\mc}[1]{\ensuremath{\mathcal{#1}}}

\makeatletter
\renewcommand{\paragraph}{%
  \@startsection{paragraph}{4}%
  {\z@}{2pt}{-0.5em}%
  {\normalfont\normalsize\bfseries}%
}
\makeatother

\newcommand{\edgeoracle}[1]{\ensuremath{\mathsf{EO}(#1)}}
\newcommand{\vertexoracle}[1]{\ensuremath{\mathsf{VO}(#1)}}
\newcommand{\true}[0]{\ensuremath{\mathsf{TRUE}}}
\newcommand{\false}[0]{\ensuremath{\mathsf{FALSE}}}

\DeclareMathOperator{\poly}{poly}

\addauthor{sb}{blue}    

\newtheorem{theorem}{Theorem}[section]
\newtheorem{lemma}[theorem]{Lemma}
\newtheorem{proposition}[theorem]{Proposition}
\newtheorem{corollary}[theorem]{Corollary}

\newtheorem{definition}[theorem]{Definition}
\newtheorem{claim}[theorem]{Claim}

\newtheorem{observation}[theorem]{Observation}
\newtheorem{remark}[theorem]{Remark}
\newtheorem{assumption}[theorem]{Assumption}

\newcommand{\etal}[0]{\textit{et al.}}

\definecolor{mygreen}{RGB}{20,120,60}
\definecolor{myred}{RGB}{200,60,60}
\definecolor{mylightgray}{RGB}{230,230,230}

\hypersetup{
     colorlinks=true,
     citecolor= mygreen,
     linkcolor= myred
}

\definecolor{quoteboxborder}{RGB}{200,200,200}

\newenvironment{quotebox}{
\par\addvspace{0.2cm}
\begin{tcolorbox}[width=\textwidth,
                  enhanced,
                  frame hidden,
                  interior hidden,
                  boxsep=0pt,
                  left=10pt,
                  right=0pt,
                  top=0pt,
                  bottom=0pt,
                  boxrule=1pt,
                  arc=0pt,
                  colback=white,
                  colframe=black,
                  borderline west={3pt}{0pt}{quoteboxborder}
                  ]
}{
\end{tcolorbox}
\smallskip
}

\newenvironment{graytbox}{
\par\addvspace{0.1cm}
\begin{tcolorbox}[width=\textwidth,
                  enhanced,
                  frame hidden,
                  boxsep=5pt,
                  left=1pt,
                  right=1pt,
                  top=2pt,
                  bottom=2pt,
                  boxrule=1pt,
                  arc=0pt,
                  colback=mylightgray,
                  colframe=black,
                  breakable
                  ]
}{
\end{tcolorbox}
}

\begin{document}

\maketitle

\thispagestyle{empty}

\renewenvironment{abstract}
 {\small
  \begin{center}
  \bfseries \abstractname\vspace{-.5em}\vspace{0pt}
  \end{center}
  \list{}{%
    \setlength{\leftmargin}{6mm}
    \setlength{\rightmargin}{\leftmargin}%
    \listparindent 1.5em%
    \itemindent    \listparindent%
  }%
  \item\relax}
 {\endlist}

\begin{abstract}
	We study the problem of estimating the size of maximum matching and minimum vertex cover in sublinear time. Denoting the number of vertices by $n$ and the average degree in the graph by $\bar{d}$, we obtain the following results for both problems:\footnote{The $\widetilde{O}(\cdot)$ notation hides $\poly\log n$ factors throughout the paper.}
	\begin{itemize}[leftmargin=15pt]
		\item A multiplicative $(2+\epsilon)$-approximation that takes $\widetilde{O}(n/\epsilon^2)$ time using adjacency list queries.
		\item A multiplicative-additive $(2, \epsilon n)$-approximation in $\widetilde{O}((\bar{d} + 1)/\epsilon^2)$ time using adjacency list queries.
		\item A multiplicative-additive $(2, \epsilon n)$-approximation in $\widetilde{O}(n/\epsilon^{3})$ time using adjacency matrix queries.
	\end{itemize}
	
	All three results are provably time-optimal up to polylogarithmic factors culminating a long line of work on these problems.
	
	\smallskip
	
	Our main contribution and the key ingredient leading to the bounds above is a new and near-tight analysis of the {\em average query complexity} of the randomized greedy maximal matching algorithm which improves upon a seminal result of Yoshida, Yamamoto, and Ito [STOC'09].
\end{abstract}

\clearpage
\thispagestyle{empty}
\setcounter{tocdepth}{3}
\tableofcontents
\clearpage

\setcounter{page}{1}
\section{Introduction}\label{sec:intro}

We study algorithms for estimating the \underline{size} of maximum matching or minimum vertex cover in general graphs. A {\em maximal} matching, which can be trivially constructed in linear time, leads to simple 2-approximations for both problems. But can we achieve the same in {\em sublinear} time?

Since sublinear-time algorithms cannot even read the whole data, it is important to specify how the input is presented. For graphs, two models have been commonly considered in the literature:
\begin{itemize}[topsep=2pt, itemsep=0pt] 
	\item \textbf{The adjacency list model:} In this model, for any vertex $v$ of its choice, the algorithm may query the degree of $v$ in the graph and, for any $1 \leq i \leq \deg(v)$, may query the $i$-th neighbor of $v$ stored in an arbitrarily ordered list.
	\item \textbf{The adjacency matrix model:} In this model, the algorithm may query, for any vertex-pair $(u, v)$ of its choice, whether or not $u$ and $v$ are adjacent in the graph.
\end{itemize}
We consider both models in this work.

Although at the first glance it may seem impossible to do much without reading the whole input, numerous sublinear-time algorithms have been designed over the years for various optimization problems. In addition to matching and vertex cover, which have been studied extensively in the area \cite{ParnasRon07,NguyenOnakFOCS08,YoshidaSTOC09,OnakSODA12,KapralovKSSODA14,ChenICALP20}, the list includes estimating the weight/size of minimum spanning tree (MST) \cite{ChazelleRT05,CzumajS09MST}, traveling salesman problem (TSP) \cite{ChenICALP20}, $k$-nearest neighbor graph \cite{CzumajS20NearestNeighbor}, graph's average degree \cite{Feige06, GoldreichR08}, as well as problems such as vertex coloring \cite{AssadiCK19}, metric linear sampling \cite{EsfandiariM18}, and many others. (This is by no means a comprehensive list of all the prior works.) For some of the classic results of the area, see the excellent survey of Czumaj and Sohler \cite{CzumajS10}.

The famous {\em randomized greedy maximal matching} (RGMM) algorithm is the basis of one of the most successful approaches for estimating the size of matching and vertex cover in sublinear time. The RGMM algorithm iterates over the edges in a {\em random} order and greedily adds each encountered edge to the matching if not already incident to a matching edge. Although it takes a linear time to run this algorithm on the whole graph, one can locally determine if a single vertex is part of the solution much faster. The key observation is the fact that an edge is in the output matching of RGMM, if and only if it has no neighboring edge earlier in the ordering that belongs to the matching. This naturally gives rise to a recursive query process, suggested first by Nguyen and Onak \cite{NguyenOnakFOCS08}, that explores the local neighborhood of a given edge or vertex, and determines if it is matched in an instance of RGMM (we formalize this in Section~\ref{sec:query}).

To utilize this local simulation of RGMM, the most common approach is to plug it into the framework of Parnas and Ron \cite{ParnasRon07}: Pick a number of random vertices in the graph, simulate the local RGMM on them, and report the fraction of them that are matched as an estimate for the fraction of vertices matched in the whole graph. The final time-complexity of the algorithm, therefore, depends on the size of the local neighborhood that one has to explore for each randomly chosen vertex, which is also known as the ``average query-complexity'' of RGMM \cite{NguyenOnakFOCS08,YoshidaSTOC09,OnakSODA12}.

Indeed, our main technical contribution is the following near-tight analysis of the average query-complexity of RGMM, which improves upon a seminal work of Yoshida, Yamamoto, and Ito~\cite{YoshidaSTOC09}.

\begin{graytbox}
\begin{theorem}[See Theorem~\ref{thm:querycomplexity} for the formal statement]\label{thm:query-intro}
	For any $n$-vertex graph with average degree $\bar{d}$, the average query-complexity of RGMM is $O(\bar{d} \cdot \log n)$.
\end{theorem}
\end{graytbox}

We note that in a complete bipartite graph with $\Theta(\bar{d})$ vertices in one part and $\Theta(n)$ vertices in the other, for $n \gg \bar{d}$, the average query-complexity is $\Omega(\bar{d})$. Hence Theorem~\ref{thm:query-intro} is almost tight. 

The result of Yoshida~\etal{}~\cite{YoshidaSTOC09}, in comparison, bounds the query-complexity of a random vertex by $O(L/n)$, where $L$ is the number of edges in the line-graph. In general, $L/n$ can be upper bounded by $O(\bar{d} \cdot \Delta)$ where $\Delta$ is the maximum degree, and there are graphs\footnote{Consider a complete bipartite graph with $\Delta + 1$ vertices in one part and $\Theta(\bar{d})$ vertices in the other part.\label{footnote:Ln}} for which $L/n = \Omega(\bar{d} \cdot \Delta)$. We elaborate more on this result of \cite{YoshidaSTOC09} and further compare it to Theorem~\ref{thm:query-intro} in Section~\ref{sec:techniques}.

\subsection{Applications of the New Query-Complexity Analysis}

Theorem~\ref{thm:query-intro} leads to a number of sublinear-time algorithms for matching and vertex cover that are provably time-optimal up to logarithmic factors. We elaborate on these applications here and also mention some other applications of Theorem~\ref{thm:query-intro}.

As before, in all the statements below (and throughout the paper) we use $n$ to denote the number of vertices, $\Delta$ to denote the maximum degree, and $\bar{d}$ to denote the average degree. 

\medskip

\paragraph{Application 1:} The first application is a multiplicative approximation in the adjacency list model.

\begin{graytbox}
\begin{theorem}\label{thm:adjlist-multiplicative}
	For any $\epsilon > 0$, there is an algorithm that with probability $1-1/\poly(n)$ reports a $(2 + \epsilon)$-approximation to the size of maximum matching and that of minimum vertex cover using $O(n) + \widetilde{O}(\Delta/\epsilon^{2})$ time and queries in the adjacency list model.
\end{theorem}
\end{graytbox}

Observe that $\Omega(n)$ queries are necessary even to distinguish an empty graph from one that includes a single edge. As such, $\Omega(n)$ time and queries are necessary for any multiplicative approximation in this model, implying that Theorem~\ref{thm:adjlist-multiplicative} is nearly time-optimal.

Theorem~\ref{thm:adjlist-multiplicative}, notably, gives the {\em first} multiplicative estimator in the literature that runs in $\widetilde{O}(n)$ time for all graphs. For a $(2+\epsilon)$-approximation, in particular, no $o(n^2)$ time algorithm was known for general graphs prior to our work. Allowing a larger $O(1)$-approximation, a recent result of \cite{KapralovSODA20} with some work leads to a $O(n + \Delta^2/\bar{d})$ time algorithm which can take $\widetilde{\Omega}(n \sqrt{n})$ time.

Interestingly, under the rather mild assumptions that $\bar{d} = \Omega(1)$ (which, e.g., holds if there are no singleton vertices) and that (estimates of) $\bar{d}$ and $\Delta$ are given, the $\Omega(n)$ lower bound is avoidable and the algorithm of Theorem~\ref{thm:adjlist-multiplicative} actually runs in $\widetilde{O}(\Delta/\epsilon^2)$ time.

\medskip

\paragraph{Application 2:} The second application of Theorem~\ref{thm:query-intro} is a multiplicative-additive approximation (see Section~\ref{sec:preliminaries} for the formal definition), also in the adjacency list model.

\begin{graytbox}
\begin{theorem}\label{thm:adjlist-additive}
	For any $\epsilon > 0$, there is an algorithm that with probability $1-1/\poly(n)$ reports a $(2, \epsilon n)$-approximation to the size of maximum matching and that of minimum vertex cover using $\widetilde{O}((\bar{d} + 1)/\epsilon^2)$ time and queries in the adjacency list query model.
\end{theorem}
\end{graytbox}

Theorem~\ref{thm:adjlist-additive} (nearly) matches an $\Omega(\bar{d} + 1)$ lower bound of Parnas and Ron \cite{ParnasRon07} that holds for any $(O(1), \epsilon n)$-approximation of maximum matching or minimum vertex cover in this model. This culminates a long line of work \cite{ParnasRon07,NguyenOnakFOCS08,YoshidaSTOC09,OnakSODA12,KapralovSODA20} on this problem that we overview next.

Parnas and Ron \cite{ParnasRon07} were the first to give a multiplicative-additive approximation for matching and vertex cover. They showed how to obtain a $(2, \epsilon n)$-approximation in $\Delta^{O(\log(\Delta /\epsilon))}$ time by simulating a distributed local algorithm for each sampled vertex. A quasi-polynomial dependence on $\Delta$, however, is unavoidable with this approach due to existing distributed lower bounds \cite{KuhnMW-JACM16}. 

The next wave of results \cite{NguyenOnakFOCS08,YoshidaSTOC09,OnakSODA12} were based on the RGMM algorithm. Yoshida~\etal{}~\cite{YoshidaSTOC09} built on the work of Nguyen and Onak \cite{NguyenOnakFOCS08} and, notably, gave the first algorithm with a polynomial-in-$\Delta$ time-complexity of $O(\Delta^4/\epsilon^2)$ for a $(2, \epsilon n)$-approximation.  Onak~\etal{}~\cite{OnakSODA12} later shaved off some of the $\Delta$ factors from the bound of \cite{YoshidaSTOC09}. Although a bound of $\widetilde{O}_\epsilon(\Delta)$ was first claimed in~\cite{OnakSODA12}, Chen, Kannan, and Khanna \cite{ChenICALP20} discovered a subtlety with a claim of \cite{OnakSODA12} that happens to be crucial for their final result. Nonetheless, as also observed by Chen~\etal{}~\cite{ChenICALP20}, the techniques developed in \cite{OnakSODA12} combined with the average query-complexity analysis of \cite{YoshidaSTOC09} discussed above, still implies an improved bound of $\widetilde{O}((\bar{d} \cdot \Delta + 1) / \epsilon^2) = \widetilde{O}(\Delta^2/\epsilon^2)$ for a $(2, \epsilon n)$-approximation.\footnote{We note that while Theorem~\ref{thm:adjlist-additive} can be seen as a fix to \cite{OnakSODA12}, our techniques are very different from the approach of \cite{OnakSODA12}. Particularly, \cite{OnakSODA12} did not claim the tight upper bound of Theorem~\ref{thm:query-intro} that we prove in this paper, but rather claimed a weaker bound of $O(\bar{d} \cdot \rho)$ where $\rho$ is the ratio of the maximum degree over the min degree. As a result, even if one manages to fix the proof of \cite{OnakSODA12}, one does not get the strong multiplicative approximation of Theorem~\ref{thm:adjlist-multiplicative}.}

Observe that all the algorithms of the literature discussed above require some upper bound on the degrees to run in sublinear-time and can take up to $\Omega(n^2)$ time for general graphs. A more recent algorithm by Kapralov~\etal{}~\cite{KapralovSODA20} has a more desirable time-complexity of $O(\Delta/\epsilon^2)$. However, it only obtains an $(O(1), \epsilon)$-approximation and as pointed out by Chen, Kannan, and Khanna \cite{ChenICALP20}:
\begingroup
\advance\leftmargini -2em
\begin{quotebox}
{\em ``Unfortunately, the constant hidden in the $O(1)$ notation {\normalfont [{\em of} \cite{KapralovSODA20}]} is very large, and efficiently obtaining a $(2, \epsilon n)$-approximation to matching size remains an important open problem''} \cite{ChenICALP20}.
\end{quotebox}
\endgroup

\noindent Theorem~\ref{thm:adjlist-additive}, in light of the lower bound of \cite{ParnasRon07}, resolves this open problem of \cite{ChenICALP20} in a strong sense.

\medskip
\paragraph{Application 3:} The third application is in the adjacency matrix model.

\begin{graytbox}
\begin{theorem}\label{thm:adjmatrix}
	For any $\epsilon > 0$, there is an algorithm that with probability $1-1/\poly(n)$ reports a $(2, \epsilon n)$-approximation to the size of maximum matching and that of minimum vertex cover using $\widetilde{O}(n/\epsilon^3)$ time and queries in the adjacency matrix query model.
\end{theorem}
\end{graytbox}

A similar bound to Theorem~\ref{thm:adjmatrix} was claimed in \cite{OnakSODA12}. However, as discussed above, the proof of \cite{OnakSODA12} had a subtlety. Chen~\etal{}~\cite{ChenICALP20} proposed a fix, but their algorithm runs in $\widetilde{O}_\epsilon(n \sqrt{n})$ time. Theorem~\ref{thm:adjmatrix} improves this latter bound by a factor of $\sqrt{n}$.

On the lower bound side, suppose that the graph is either a random perfect matching or is empty. It is not hard to see that distinguishing the two cases requires $\Omega(n)$ queries to the adjacency matrix. As such, $\Omega(n)$ queries are necessary for any multiplicative-additive approximation in the adjacency-matrix model, implying that Theorem~\ref{thm:adjmatrix} is nearly time-optimal. Note that distinguishing an empty graph from one that includes only one edge requires $\Omega(n^2)$ queries in the adjacency-matrix model. This implies that, unlike Theorem~\ref{thm:adjlist-multiplicative} for the adjacency-list model, no non-trivial multiplicative approximation can be obtained in the adjacency-matrix model.

To prove Theorem~\ref{thm:adjmatrix}, in addition to Theorem~\ref{thm:query-intro}, we give a new reduction from the adjacency matrix model to the adjacency list model that, unlike the reduction of \cite{OnakSODA12}, does not lead to parallel edges or self-loops. This is crucial for our result; see Section~\ref{sec:adjmatrix} for details.

\medskip
\paragraph{Other Applications -- TSP:} Chen~\etal{}~\cite{ChenICALP20} gave sublinear time algorithms for estimating the size of variants of TSP such as graphic TSP in the adjacency matrix model. Their algorithms utilize a matching size estimator as a subroutine. Plugging Theorem~\ref{thm:adjmatrix} into their framework implies that:

\begin{corollary}[of using Theorem~\ref{thm:adjmatrix} in the algorithms of \cite{ChenICALP20}]\label{cor:TSP}
	There is an $\widetilde{O}(n)$-time randomized algorithm that estimates the cost of graphic TSP to within a factor of (27/14).
\end{corollary}

Instead of Theorem~\ref{thm:adjmatrix}, Chen~\etal{}~\cite{ChenICALP20} used their $\widetilde{O}(n \sqrt{n})$-time matching size estimator in their algorithms. As a result, their algorithms were slower by a factor of $\sqrt{n}$ than those in Corollary~\ref{cor:TSP}.

\medskip
\paragraph{Other Applications -- AMPC:} The {\em adaptive massively parallel computations} (AMPC) model was first introduced by \cite{BehnezhadDELMS19SPAA} and has been further studied by \cite{BehnezhadDELMS20VLDB,CharikarMT20}. While the precise definition of the AMPC model is out of the scope of this paper, it augments the standard MPC model of \cite{KarloffSV10} with a distributed hash table. Combined with the techniques developed for the maximal independent set problem in \cite{BehnezhadDELMS19SPAA}, Theorem~\ref{thm:query-intro} implies an $O(1)$-round AMPC algorithm for maximal matching using a strongly sublinear space of $O(n^\delta)$ per machine (for any constant $\delta > 0$) and an optimal total space of $\widetilde{O}(m)$ in $m$-edge graphs. This has to be compared with an $O(\log \log n)$-round algorithm presented in \cite{BehnezhadDELMS20VLDB} that has the same space complexity.

The essence of the improved AMPC algorithm mentioned above is the surprising fact, implicit in Theorem~\ref{thm:query-intro}, that if one starts the local simulation of RGMM from {\em every} vertex in the graph all in parallel, then the expected sum of the query-complexities of all these parallel calls, or equivalently the ``total work'' of the algorithm, is $n \cdot O(\bar{d} \log n) = O(m \log n)$, i.e., near-linear in the input size!

\subsection{Our Techniques \& Background on the Query-Complexity of RGMM}\label{sec:techniques}

As discussed above, our main technical contribution is the upper bound of Theorem~\ref{thm:query-intro} on the average query-complexity of the randomized greedy maximal matching algorithm. In this section, we provide more context about this result, further compare it to the previous bounds, and give an overview of our techniques and new insights for proving it.

Let us start by giving a slightly more formal definition of the local simulation of RGMM due to \cite{NguyenOnakFOCS08}. Suppose that we would like to define an oracle, $\edgeoracle{e, \pi}$, to determine whether a given edge $e$ belongs to the matching $\GMM{G, \pi}$ returned by the greedy maximal matching algorithm processing the edges in the order of $\pi$. Nguyen and Onak \cite{NguyenOnakFOCS08} observed that $e$ belongs to $\GMM{G, \pi}$ if and only if no edge incident to $e$ with a lower $\pi$-rank than $e$ belongs to $\GMM{G, \pi}$. Therefore, one can recursively run $\edgeoracle{e', \pi}$ on all incident edges $e'$ to $e$ with $\pi(e') < \pi(e)$. If none of these neighbors join $\GMM{G, \pi}$, we know for sure that $e \in \GMM{G, \pi}$, and otherwise $e \not\in \GMM{G, \pi}$.

The next crucial observation is that as soon as we find just one neighboring edge $e'$ of $e$ that is in the matching, we can terminate $\edgeoracle{e, \pi}$ and report $e \not\in \GMM{G, \pi}$ without going through the rest of the neighbors of $e$. Intuitively, this ``pruning'' should have a significant effect on the average query-complexity due to the recursive nature of the oracle. But how can we prove it? An immediate problem is that analyzing the expected query-complexity for a {\em given} edge $e$ seems to be challenging. In fact, obtaining a $\poly(\Delta, \log n)$ bound for this problem remains a major open question in the study of local computation algorithms (LCA's) \cite{AlonRVX12,GhaffariUittoSODA19}.

In a beautiful paper, Yoshida, Yamamoto, and Ito \cite{YoshidaSTOC09} got around this challenge and proved a $\poly(\Delta)$ upper bound on the {\em average} query-complexity with a global analysis. Namely, instead of bounding the expected number of queries that an edge $e$ generates, they showed that for any edge $e=(u, v)$ the expected number of edges in the graph whose query process reaches $e$ is at most $O(\deg(u)+\deg(v))$.\footnote{The analysis of \cite{YoshidaSTOC09} proceeds on the line-graph and $\deg(u)+\deg(v)$ is the degree of $e$ in the line-graph.} This way, the expected sum of all queries starting from all the edges in the graph, can be upper bounded by $\sum_{(u, v) \in E} O(\deg(u) + \deg(v)) = O(L)$, where $L$ is the number of edges in the line-graph. This implies that for an {\em edge} chosen uniformly at random, the expected query-complexity can be upper bounded by $O(L/m)$. Although it is not immediate, we note that essentially the same proof also implies that for a {\em vertex} chosen uniformly at random, the expected query-complexity is $O(L/n)$ which is $O(\bar{d} \cdot \Delta)$ and can also be as large as $\Omega(\bar{d} \cdot \Delta)$ (see Footnote~\ref{footnote:Ln}).

Our new analysis builds on some of the ideas introduced by Yoshida~\etal{}~\cite{YoshidaSTOC09}. Similar to their work and, crucially, instead of directly analyzing the number of recursive calls {\em out} of an edge or a vertex, we analyze the recursive calls made {\em to} an edge or a vertex. There are, however, several crucial differences between the two approaches that leads to our near-tight analysis in Theorem~\ref{thm:query-intro}.

We prove that for any edge $e=(u, v)$ in the graph, the expected number of starting queries that reach $e$ can be upper bounded by $O(\log n)$. This improves over the previous upper bound of $O(\deg(u) + \deg(v))$ by Yoshida~\etal{}~\cite{YoshidaSTOC09} and is the key component of our analysis. 

Let us fix edge $e$ and overview how we prove the bound of $O(\log n)$ on the total number of queries to $e$, from all possible starting points. To prove this upper bound, if in a permutation $\pi$ we happen to query $e$ for $q$ times, we charge $q$ other permutations $\sigma_1, \ldots, \sigma_q$. Observe that by a simple double counting argument, the expected number of times that a random permutation $\sigma \in \Pi$ is charged by all other permutations combined, is exactly equal to the expected number of queries to $e$ in a random permutation $\pi$. As such, if we prove an upper bound of $T$ on the number of times that each permutation is charged, then we get that $e$ is queried by at most $T$ starting points in expectation. Unfortunately, however, we remark that there will be permutations that get charged up to even $\Omega(n)$ times with our charging method.

To get past this hurdle, we draw a novel connection to an orthogonal line of work on bounding the {\em parallel depth} of RGMM pioneered by Blelloch, Fineman, and Shun \cite{BlellochSPAA} whose bounds were tightened by Fischer and Noever \cite{FischerTALG}, showing that these algorithms have parallel depth $O(\log n)$ w.h.p. By carefully analyzing the structure of RGMM queries, we show that if a permutation $\sigma$ is charged $b$ times with our charging method, then for any $1 \leq i \leq b$, at least $i$ of these charges should be from permutations with parallel depth at least $\Omega(b - i)$. In particular, if a permutation $\sigma$ is charged $\omega(\log n)$ times, most of these charges must be from those ``unlikely'' permutations that have a high parallel depth $\omega(\log n)$. With a few additional ideas (particularly by caching previous queries) we show that the queries of these unlikely permutations can be effectively bounded, and get an $O(\log n)$ upper bound on the number of charges to an average permutation $\sigma$.

\vspace{-0.2cm}
\section{Preliminaries}\label{sec:preliminaries}

Unless otherwise stated, we use $n$ to denote the number of vertices, $m$ to denote the number of edges, $\Delta$ to denote the maximum degree, and $\bar{d}$ to denote the average degree in the original graph. The graphs considered in this paper are unweighted and simple. We use $\mu(G)$ to denote the size of the maximum matching in $G$ and use $\nu(G)$ to denote the size of the minimum vertex cover in $G$.  For any $U \subseteq V$, we use $G[U]$ to denote the induced subgraph of $G=(V, E)$ on $U$, which is the subgraph obtained by removing all vertices in $V \setminus U$ (along with their edges) from $G$. For any positive integer $k$ we use $[k]$ to denote the set $\{1, \ldots, k\}$.

As standard in the literature, for $\alpha \geq 1$ and $0 \leq \epsilon \leq 1$, we say an estimate $\widetilde{\mu}(G)$ for $\mu(G)$ and an estimate $\widetilde{\nu}(G)$ for $\nu(G)$ provide ``multiplicative-additive'' $(\alpha, \epsilon n)$-approximations if 
$$
\frac{\mu(G)}{\alpha} - \epsilon n \leq \widetilde{\mu}(G) \leq \mu(G) \qquad \text{and} \qquad
\nu(G) \leq \widetilde{\nu}(G) \leq \alpha \nu(G) + \epsilon n.
$$
Note that a standard multiplicative $\alpha$-approximation is equivalent to an $(\alpha, 0)$-approximation.

We will use the following standard version of the Chernoff bound:

\begin{proposition}[Chernoff Bound]\label{prop:chernoff}
	Let $X_1, \ldots, X_n$ be independent Bernoulli random variables and define $X := \sum_{i=1}^n X_i$. For any $\lambda > 0$,
	$
		\Pr[|X - \E[X]| \geq \lambda] \leq 2 \exp\left(- \frac{\lambda^2}{3 \E[X]}\right).
	$
\end{proposition}

\section{Average Query-Complexity of RGMM}\label{sec:query}

Throughout this section, we let $G=(V, E)$ be a graph with $n$ vertices, $m$ edges, and average degree $\bar{d} = 2m/n$. We also let $\Pi$ be the set of all $m!$ permutations over the edge-set $E$. 

\begin{definition}[Greedy Maximal Matching]\label{def:RGMM}
For a permutation $\pi \in \Pi$, let $\GMM{G, \pi}$ be a maximal matching of graph $G=(V, E)$ obtained by iterating over the edge-set $E$ in the order of $\pi$ and greedily adding each encountered edge that is feasible to the matching (an edge is feasible if it is not already incident to a matching edge).
\end{definition}

We are interested in the behavior of the GMM algorithm for a \underline{random permutation} $\pi$ chosen uniformly from $\Pi$. Particularly, our goal is to determine if a vertex belongs to the produced matching. Of course one can run the algorithm and answer this in $\Theta(m + n)$ time. But as we will see, better bounds can be achieved via the following local procedure. Similar oracles were analyzed by \cite{NguyenOnakFOCS08,YoshidaSTOC09,OnakSODA12}. But we emphasize that our edge oracle is slightly more efficient than previous ones.

\smallskip

\begin{algorithm}[H]
\label{alg:vertexoracle}
\caption{The vertex oracle $\vertexoracle{v, \pi}$. Determines if vertex $v$ is matched in $\GMM{G, \pi}$.}
	Let $e_1=(v, u_1), \ldots, e_k=(v, u_k)$ be the edges incident to $v$ with $\pi(e_1) < \ldots < \pi(e_k)$.
	
	\For{$i$ in $1 \ldots k$}{
		\lIf{$\edgeoracle{e_i, u_i, \pi} = \true$}{\Return \true}
	}
	\Return \false
\end{algorithm}

\smallskip

\begin{algorithm}[H]
\label{alg:edgeoracle}
\caption{The edge oracle $\edgeoracle{e, u, \pi}$ where $u$ is an endpoint of edge $e$.}

	\lIf{we have already computed $\edgeoracle{e, u, \pi}$}{\Return the computed answer.}\label{line:cache}
	
	Let $e_1=(u, w_1), \ldots, e_k=(u, w_k)$ be the edges adjacent to $u$ such that $\pi(e_i) < \pi(e)$ and $\pi(e_1) < \ldots < \pi(e_k)$.
	
	\For{$i$ in $1 \ldots k$}{
		\lIf{$\edgeoracle{e_i, w_i, \pi} = \true$}{\Return \false}
	}
	\Return \true
\end{algorithm}

\smallskip

\begin{remark}\label{rem:cache}
	In Line~\ref{line:cache} of Algorithm~\ref{alg:edgeoracle} we ``cache'' previous edge queries. This ensures that the total recursive calls to the edge oracle is bounded by $\poly(n)$ with probability 1 (see Observation~\ref{obs:detupperbound}). This is the only implication of caching that we use in our analysis.
\end{remark}

\begin{remark}
	Our edge oracle differs from those of \cite{NguyenOnakFOCS08,YoshidaSTOC09,OnakSODA12} in that the oracle calls generated by $\edgeoracle{e, u, \pi}$ are only for those edges incident to $u$ and not the other endpoint of $e$. This helps in arguing that the stack of recursive calls to the edge oracle form a path in $G$.
\end{remark}

\begin{observation}\label{obs:queries-correctness}
	It holds for the oracles above that:
	\begin{itemize}[itemsep=0pt,topsep=0pt]
		\item For any vertex $v \in V$, $\vertexoracle{v, \pi} = \true$ iff $v$ is matched in $\GMM{G, \pi}$.
		\item For any edge $e=(u, w)$, if $\edgeoracle{e, u, \pi}$ is (recursively) called while evaluating some vertex oracle $\vertexoracle{v, \pi}$, then we have $\edgeoracle{e, u, \pi} = \true$ iff $e \in \GMM{G, \pi}$.
	\end{itemize} 
\end{observation}
\begin{proof}
	The second property immediately implies the first as \vertexoracle{v, \pi} calls the edge oracle on the neighbors of $v$ until finding a matching edge.
	
	We prove the second property by an induction on $\pi(e)$, so assume that this statement holds for all edges with a lower rank than $\pi(e)$. Observe that $\edgeoracle{e, u, \pi}$ is either directly called by $\vertexoracle{v, \pi}$ which in that case $v = w$, or it is directly called by $\edgeoracle{f, w, \pi}$ for some edge $f$ incident to $w$. In both cases, by the description of the oracles, the function that calls \edgeoracle{e, u, \pi} should first call \edgeoracle{e', u', \pi} for any edge $e'=(w, u')$ incident to $w$ that has a lower rank than $\pi(e)$, and that all these calls should have returned \false{}. By the induction hypothesis, this implies that no such $e'$ belongs to \GMM{G, \pi}. Therefore, if an edge of \GMM{G, \pi} prevents $e$ from joining \GMM{G, \pi}, it must be incident to $e$ from its other endpoint $u$, and should have a lower rank than $\pi(e)$. Thus, it is sufficient that the edge oracle \edgeoracle{e, u, \pi} only goes over such edges of $u$ (for which the statement holds by the induction hypothesis), returning \true{} iff none of these calls return \true{}.
\end{proof}

Let $T(v, \pi)$ denote the number of recursive calls to the edge oracle $\edgeoracle{\cdot, \cdot, \pi}$ over the course of answering $\vertexoracle{v, \pi}$. We note that for some edge $e$, the edge oracle may be called multiple times during the execution of $\vertexoracle{v, \pi}$ and we count all of these in $T(v, \pi)$, though only the first call to $\edgeoracle{e, \pi}$ may generate new recursive calls as the rest of the calls use the cached value.

Our main result in this section is as follows:

\begin{theorem}\label{thm:querycomplexity}
	For a vertex $v$ chosen uniformly at random from $V$ and for a permutation $\pi$ chosen uniformly at random from $\Pi$ and independently from $v$,
	$$
		\E_{v \sim V, \pi \sim \Pi}[T(v, \pi)] = O(\bar{d} \cdot \log n).
	$$
\end{theorem}

For any edge $e \in E$ and a vertex $v \in V$, we use $Q(e, v, \pi)$ to denote the number of times that $\edgeoracle{e, \cdot, \pi}$ is called during the execution of $\vertexoracle{v, \pi}$, and we define $Q(e, \pi) := \sum_{v \in V} Q(e, v, \pi)$. The following immediately follows from the definition:

\begin{observation}\label{obs:TvsQ}
	For any $\pi \in \Pi$ and any $v \in V$, $T(v, \pi) = \sum_{e \in E} Q(e, v, \pi)$.
\end{observation}

The next bound holds because we cache query results as discussed in Remark~\ref{rem:cache}.

\begin{observation}\label{obs:detupperbound}
	For every edge $e = \{a, b\}$ and every $\pi \in \Pi$, $Q(e, \pi) \leq O(n^2)$.
\end{observation}
\begin{proof}
	Take an arbitrary vertex $v \in V$. Over the course of answering $\vertexoracle{v, \pi}$, we either call $\edgeoracle{e, \cdot, \pi}$ directly by the vertex oracle (at most once) or during the execution of $\edgeoracle{f, \cdot, \pi}$ for some edge $f$ incident to $e$. The key observation is that $\edgeoracle{f, \cdot, \pi}$ generates new recursive calls only the first time that it is called (due to caching). Hence, in total $\edgeoracle{e, \cdot, \pi}$ is called at most $(\deg_G(a)-1) + (\deg_G(b)-1) + 1 \leq 2n - 1$ times while answering $\vertexoracle{v, \pi}$. This means $Q(e, v, \pi) \leq 2n-1$. Since $Q(e, \pi) = \sum_v Q(e, v, \pi)$, we get $Q(e, \pi) \leq \sum_v (2n-1) = n (2n-1) = O(n^2)$.
\end{proof}

The main technical part is the proof of Lemma~\ref{lem:querytoedge} below. Intuitively, the lemma states that if we run the vertex oracle on every vertex $v \in V$ all in parallel (in a way that the cashed values stored for one parallel call are not used for another), then in expectation over $\pi$, the total number of times that edge oracle $\edgeoracle{e, \cdot, \pi}$ is called can be bounded by $O(\log n)$.

\begin{lemma}\label{lem:querytoedge}
	For any edge $e$, $\E_\pi[Q(e, \pi)] = O(\log n).$
\end{lemma}

Lemma~\ref{lem:querytoedge} easily implies Theorem~\ref{thm:querycomplexity} as described next.

\begin{proof}[Proof of Theorem~\ref{thm:querycomplexity} via Lemma~\ref{lem:querytoedge}.]
	It holds that
	\begin{flalign}
		\nonumber \E_\pi \left[\sum_{v \in V} T(v, \pi) \right] &\stackrel{\text{Obs~\ref{obs:TvsQ}}}{=} \E_\pi \left[\sum_{v \in V} \sum_{e \in E} Q(e, v, \pi) \right] = \E_\pi \left[\sum_{e \in E} \sum_{v \in V} Q(e, v, \pi) \right] = \E_\pi \left[\sum_{e \in E} Q(e, \pi) \right]\\
		&= \sum_{e \in E} \E_\pi \left[Q(e, \pi) \right] \stackrel{\text{Lemma~\ref{lem:querytoedge}}}{=} \sum_{e \in E} O(\log n) = O(m \log n).\label{eq:hugdc19823123}
	\end{flalign}
	Therefore, for a vertex $v \in V$ that is chosen uniformly at random from $V$, we get
	$$
		\E_{v,\pi} [T(v, \pi)] = \frac{1}{n} \cdot \E_\pi \left[\sum_{v \in V} T(v, \pi) \right] \stackrel{(\ref{eq:hugdc19823123})}{=} \frac{1}{n} \cdot O(m \log n) = O(\bar{d} \cdot \log n).\qedhere
	$$
\end{proof}

The rest of this section is devoted to proving Lemma~\ref{lem:querytoedge}.

\paragraph{Query paths:} Let $S_{v, \pi}$ be the stack of recursive calls to the edge oracle during the execution of $\vertexoracle{v, \pi}$. Namely, whenever the edge oracle $\edgeoracle{e, u, \pi}$ is called on some edge $e$ we push $e$ to the stack, and pop $e$ when the value of $\edgeoracle{e, u, \pi}$ is determined. Let $P=(e_1, \ldots, e_k)$ be the ordered set of edges inside $S_{v, \pi}$ at some point during the execution of $\vertexoracle{v, \pi}$ with $e_k$ being the last edge pushed into the stack. One can verify from the description of the vertex and edge oracles that $P$ must be a path in graph $G$ with $v$ being one of its endpoints (particularly $v \in e_1$). Now if we make the edges in $P$ directed such that $\vec{P} = (\vec{e_1}, \ldots, \vec{e_k})$ is a directed path starting from $v$, we call $\vec{P}$ a $(v, \pi)$-query-path. With this definition, $T(v, \pi)$ is essentially the total number of $(v, \pi)$-query-paths.

Let us take an edge $e = \{a, b\} \in E$ and let $\vec{e} = (a, b)$ be the same edge made directed from $a$ to $b$. We define $\mc{Q}(\vec{e}, v, \pi)$ to be the set of all $(v, \pi)$-query-paths that end at $\vec{e}$ (in this precise direction) and define $\mc{Q}(\vec{e}, \pi) := \bigcup_{v \in V} \mc{Q}(\vec{e}, v, \pi)$. Furthermore, we define $Q(\vec{e}, v, \pi) := |\mc{Q}(\vec{e}, v, \pi)|$ and define $Q(\vec{e}, \pi) := \sum_{v \in V} Q(\vec{e}, v, \pi) = |\mc{Q}(\vec{e}, \pi)|$ (the latter equality follows since $\mc{Q}(\vec{e}, v, \pi)$ and $\mc{Q}(\vec{e}, u, \pi)$ are disjoint for $u \not= v$).

\begin{observation}\label{obs:directed}
	Let $e = \{a, b\}$, $\vec{e} = (a, b)$, and $\cev{e} = (b, a)$. We have $Q(e, \pi) = Q(\vec{e}, \pi) + Q(\cev{e}, \pi)$.
\end{observation}
\begin{proof}
	Recall from the definition that $Q(e, v, \pi)$ is the number of times that edge oracle $\edgeoracle{e, \cdot, \pi}$ is called during the execution of $\vertexoracle{v, \pi}$. Every time that we call $\edgeoracle{e, \cdot, \pi}$, we add $e$ to the stack $S_{v, \pi}$ and thus get exactly one $(v, \pi)$-query-path that ends at $e$ either in the direction of $\vec{e}$ or $\cev{e}$. This implies that $Q(e, v, \pi) = Q(\vec{e}, v, \pi) + Q(\cev{e}, v, \pi)$. The observation follows since by definition $Q(\vec{e}, \pi) = \sum_v Q(\vec{e}, v, \pi)$, $Q(e, \pi) = \sum_v Q(e, v, \pi)$, and $Q(\cev{e}, \pi) = \sum_v Q(\cev{e}, v, \pi)$.
\end{proof}

In what follows, we prove the following lemma which easily implies Lemma~\ref{lem:querytoedge}.

\begin{lemma}\label{lem:querytodirectededge}
	For any arbitrarily directed edge $\vec{e} = (a, b)$, $\E_\pi[Q(\vec{e}, \pi)] = O(\log n)$.
\end{lemma}

\begin{proof}[Proof of Lemma~\ref{lem:querytoedge} via Lemma~\ref{lem:querytodirectededge}]
	By Observation~\ref{obs:directed}, Lemma~\ref{lem:querytodirectededge} implies Lemma~\ref{lem:querytoedge} since for any edge $e = \{u, v\}$,
$
	\E_\pi[Q(e, \pi)] = \E_\pi[Q(\vec{e}, \pi)] + \E_\pi[Q(\cev{e}, \pi)] = O(\log n) + O(\log n) = O(\log n).\qedhere
$
\end{proof}

We now turn to prove Lemma~\ref{lem:querytodirectededge} for edge $\vec{e}$ that is fixed for the rest of the proof. 

First, given a permutation $\pi \in \Pi$ and a directed path $\vec{P} = (\vec{e_1}, \ldots, \vec{e_k})$, we define $\phi(\pi, \vec{P}) \in \Pi$ to be another permutation $\sigma \in \Pi$ constructed as:
\begin{flalign*}
	(\sigma(e_1), \ldots, \sigma(e_{k-1}), \sigma(e_k)) &:= (\pi(e_2), \ldots, \pi(e_k), \pi(e_1)), \text{ and}\\
	 \sigma(e') &:= \pi(e')  \qquad \forall e' \not\in \vec{P}.
\end{flalign*}
That is, $\phi(\pi, \vec{P})$ is obtained by rotating the ranks of $\pi$ along $\vec{P}$ in the reverse direction.

Now we construct a bipartite graph $H = H(\vec{e})$ with two parts $A$ and $B$ such that $|A| = |B| = |\Pi|$. Each permutation $\pi \in \Pi$ over the edge-set of $G$ has one corresponding vertex $\pi_A$ in $A$ and one corresponding vertex $\pi_B$ in $B$. Furthermore, for any permutation $\pi \in \Pi$ and any query-path $\vec{P}=(\vec{e_1}, \ldots, \vec{e_k} = \vec{e}) \in \mathcal{Q}(\vec{e}, \pi)$, we connect vertex $\pi_A \in A$ to vertex $\sigma_B \in B$ corresponding to permutation $\sigma := \phi(\pi, \vec{P})$.

Graph $H$ is particularly constructed such that for each permutation $\pi \in \Pi$ the degree of vertex $\pi_A \in A$ equals $Q(\vec{e}, \pi)$. Namely:

\begin{observation}\label{obs:degreesarequeries}
	For any permutation $\pi \in \Pi$, $\deg_H(\pi_A) = Q(\vec{e}, \pi)$. 
\end{observation}
\begin{proof}
	By construction there is a one-to-one mapping between the edges of $\pi_A$ and query paths $\vec{P} \in \mc{Q}(\vec{e}, \pi)$. Thus $\deg_H(\pi_A) = |\mc{Q}(\vec{e}, \pi)| = Q(\vec{e}, \pi)$.
\end{proof}

Therefore to prove Lemma~\ref{lem:querytodirectededge} that $\E_\pi[Q(\vec{e}, \pi)] = O(\log n)$, it suffices to bound the average degree of graph $H$ by $O(\log n)$. In other words, it suffices to show that for a vertex $\pi_A \in A$ chosen uniformly at random, $\E_{\pi_A \sim A}[\deg_H(\pi_A)] = O(\log n)$. This is our plan for the rest of the proof.

For some large enough constant $c \geq 1$ and parameter $\beta = c \log n$, we partition $\Pi$ into two subsets of {\em likely} permutations $L$ and {\em unlikely} permutations $U$ as follows:
\begin{flalign}
	L := \left\{ \pi \in \Pi \,\,\Big\vert\,\, \max_{P \in \mc{Q}(\vec{e}, \pi)} |P| \leq \beta \right\}, \qquad U := \Pi \setminus L.\label{eq:defLU}
\end{flalign}
In words, if all query-paths in $\mc{Q}(\vec{e}, \pi)$ have length $\leq \beta$ then $\pi \in L$, and otherwise $\pi \in U$. We use $A_L$ (resp. $A_U$) to denote the set of vertices in part $A$ of graph $H$ that correspond to permutations in $L$ (resp. $U$).

We use the next two lemmas to bound the number of edges connected to $A_U$ and $A_L$ respectively.

\begin{lemma}\label{lem:AL}
	Any vertex $\sigma_B \in B$ has at most $\beta$ neighbors in $A_L$.
\end{lemma}

\begin{lemma}\label{lem:AU}
	If constant $c$ is large enough, $|A_U| \leq m! / n^2$.
\end{lemma}

Let us first see how Lemmas~\ref{lem:AL} and \ref{lem:AU} prove Lemma~\ref{lem:querytodirectededge} (and thus  Theorem~\ref{thm:querycomplexity} as discussed before). We then turn to prove these two claims.

\begin{proof}[Proof of Lemma~\ref{lem:querytodirectededge} via Lemmas~\ref{lem:AL} and \ref{lem:AU}.]
	We first upper bound the size of the edge-set $E(H)$ of $H$ by $O(m! \log n)$. Since each vertex $\pi_A \in A$ has degree $O(n^2)$ by Observation~\ref{obs:detupperbound}, and that by Lemma~\ref{lem:AU}, $|A_U| \leq m!/n^2$, the number of edges connected to $A_U$ is $\frac{m!}{n^2} \cdot O(n^2) = O(m!)$. Moreover, by Lemma~\ref{lem:AL}, each vertex $\sigma_B \in B$ has at most $\beta$ neighbors in $A_L$. Since $H$ is bipartite and every edge of $A_L$ goes to $B$, the total number of edges connected to $A_L$ can be upper bounded by $|B| \cdot \beta = m! \cdot c \log n = O(m! \log n)$. Since $A_L$ and $A_U$ partition $A$ and that every edge of $H$ has one endpoint in $A$, we get $|E(H)| = O(m! \log n) + O(m!) = O(m! \log n)$.
	
	Now if we pick a vertex $\pi_A$ from $A$, it has expected degree $\frac{|E(H)|}{|A|} = \frac{O(m! \log n)}{m!} = O(\log n)$. By Observation~\ref{obs:degreesarequeries}, this implies $\E_{\pi \in \Pi}[Q(\vec{e}, \pi)] = \E_{\pi_A \sim A}[\deg_H(\pi_A)] = O(\log n)$.
\end{proof}

\subsection{Proof of Lemma~\ref{lem:AL}}\label{sec:AL}

We prove the following statement which we show suffices to prove Lemma~\ref{lem:AL}.

\begin{claim}\label{cl:nobranch}
	Let $\pi, \pi' \in \Pi$ with $\pi \not= \pi'$, let $\vec{P}$ and $\vec{P}'$ respectively be $(v, \pi)$- and $(v', \pi')$-query-paths both ending at $\vec{e}$ for some $v, v' \in V$, and suppose that $\phi(\pi, \vec{P}) = \phi(\pi', \vec{P}') = \sigma$. Then $|\vec{P}| \not= |\vec{P}'|$.
\end{claim}

Let us first see how Claim~\ref{cl:nobranch} suffices to prove Lemma~\ref{lem:AL}:

\begin{proof}[Proof of Lemma~\ref{lem:AL} via Claim~\ref{cl:nobranch}]
	Let us take an arbitrary vertex $\sigma_B$ in part $B$ of graph $H$. For any edge $\{\pi_A, \sigma_B\}$ in $H$, by construction of $H$ there must be a unique $(v, \pi)$-query-path $\vec{P}$ such that $\phi(\pi, \vec{P}) = \sigma$ where here $\pi, \sigma \in \Pi$ are the permutations corresponding to $\pi_A$ and $\sigma_B$ respectively. Let us define the label $\chi(\pi_A, \sigma_B)$ of edge $\{ \pi_A, \sigma_B \}$ to be the length $|\vec{P}|$ of this query-path $\vec{P}$. Claim~\ref{cl:nobranch} essentially implies that all the edges of $\sigma_B$ receive different labels. On the other hand, if $\pi_A$ is a neighbor of $\sigma_B$ and $\pi_A \in A_L$, then by definition (\ref{eq:defLU}) of set $L$, $\chi(\pi_A, \sigma_B) \leq \beta$. The uniqueness of the integer labels and the upper bound of $\beta$ imply together that $\sigma_B$ can have at most $\beta$ neighbors in $A_L$. The proof of Lemma~\ref{lem:AL} is thus complete.
\end{proof}

We now turn to prove Claim~\ref{cl:nobranch}.

Let $\pi, \pi', \vec{P}, \vec{P}', v, v', \sigma$ be as defined in Claim~\ref{cl:nobranch} and assume for contradiction that $|\vec{P}| = |\vec{P}'|$. First observe that if in addition to their lengths, the edges traversed by the two paths $\vec{P}, \vec{P}'$ are also the same  i.e., $\vec{P} = \vec{P}'$, then by definition of the mapping $\phi$, we have $\phi(\pi, \vec{P}) = \phi(\pi', \vec{P}')$ iff $\pi = \pi'$ which contradicts the assumption $\pi \not= \pi'$ of Claim~\ref{cl:nobranch}. So let us assume that $\vec{P} \not= \vec{P}'$ but $|\vec{P}| = |\vec{P}'|$. This implies that the two paths must ``branch'' at least once. 

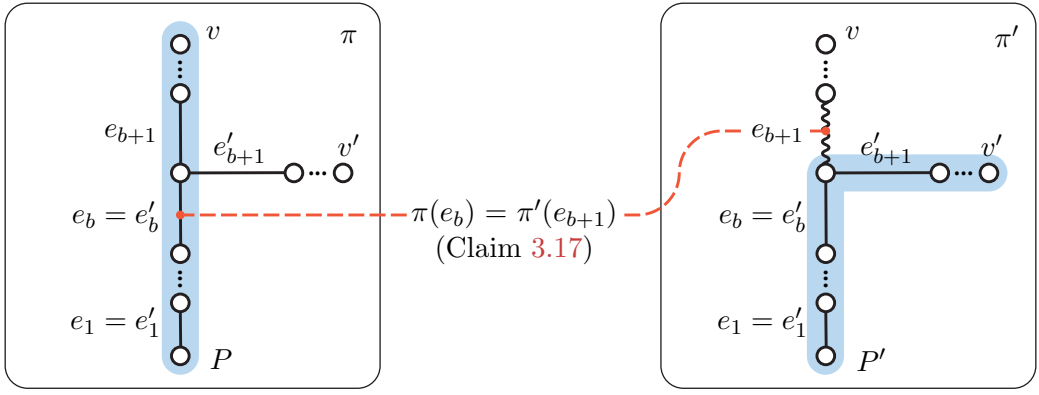
\begin{figure}[t]
\begin{tikzpicture}
    \node[anchor=south west,inner sep=0] at (0,0) {\includegraphics{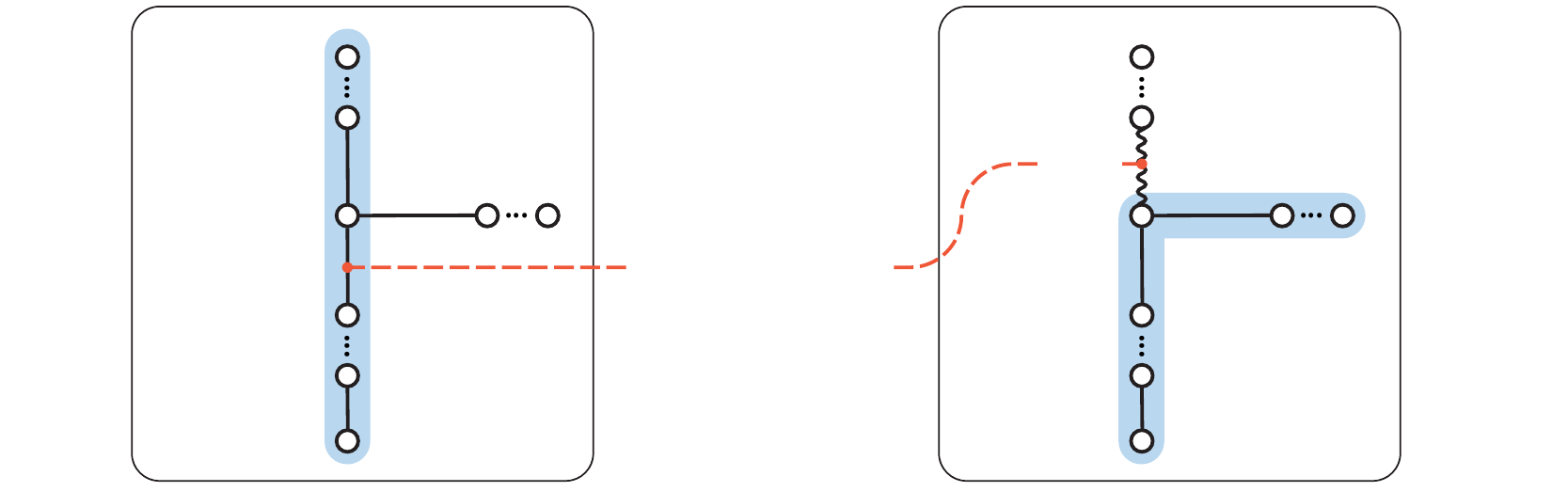}};
    \node at (2.9,1) {$e_1 = e'_1$};
    \node at (2.9,2.4) {$e_b = e'_b$};
	\node at (11.5,1) {$e_1 = e'_1$};
    \node at (11.5,2.4) {$e_b = e'_b$};
    \node at (3.1,3.5) {$e_{b+1}$};
    \node at (11.7,3.5) {$e_{b+1}$};
	\node at (13.15,3.3) {$e'_{b+1}$};
	\node at (4.55,3.3) {$e'_{b+1}$};
	\node at (4.3, 0.5) {$P$};
	\node at (12.95, 0.5) {$P'$};
	\node at (6, 4.8) {$\pi$};
	\node at (14.75, 4.8) {$\pi'$};
	\node at (4.2, 4.85) {$v$};
	\node at (12.7, 4.85) {$v$};
	\node at (6, 3.35) {$v'$};
	\node at (14.55, 3.35) {$v'$};
	\node[text width=3cm, align=center] at (8.2, 2.2) {$\pi(e_b) = \pi'(e_{b+1})$\\ (Claim~\ref{cl:theoun12389})};
\end{tikzpicture}
\caption{On the right hand side, we have permutation $\pi'$ and query-path $P'$ is highlighted. On the left hand side, we have permutation $\pi$ and path $P$ is highlighted. Our arguments essentially show that $e_{b+1}$ must belong to matching $\GMM{G, \pi'}$. We also show that $\pi'(e_{b+1}) < \pi'(e'_b)$. Therefore, $\edgeoracle{e'_{b+1}, \cdot, \pi'}$ should terminate (returning $\false$) before calling $\edgeoracle{e'_b, \cdot, \pi'}$. This means $P'$ cannot be a valid query-path in $\pi'$ and this is our desired contradiction which proves Claim~\ref{cl:nobranch}.}
\label{fig:piandpi'}
\end{figure}

It would be convenient to define $\cev{P} = (\cev{e}_1, \cev{e}_2, \ldots, \cev{e}_k)$ and $\cev{P}' = (\cev{e}'_1, \cev{e}'_2, \ldots, \cev{e}'_k)$ to be respectively the same as $\vec{P} = (\vec{e_k}, \ldots, \vec{e_1})$ and $\vec{P}' = (\vec{e_k}', \ldots, \vec{e_1}')$ except that their edges are traversed in the reverse direction. Since both $\vec{P}$ and $\vec{P}'$ end at $\vec{e}$ (by definition of query-paths), both $\cev{P}$ and $\cev{P}'$ must start from $\cev{e} = \cev{e_1} = \cev{e}_1'$. Let $(b + 1)$ be the smallest index where $\cev{e}_{b+1} \not= \cev{e}'_{b+1}$. That is, we have $e_1 = e'_1, \ldots, e_b = e'_b$ and $e_{b+1} \not= e'_{b+1}$ (see Figure~\ref{fig:piandpi'}). Observe that $b + 1 \geq 2$ since $\cev{e}_1 = \cev{e}'_1 = \cev{e}$.

Let us make the following assumption that comes without loss of generality  as we have not distinguished $\pi$ and $\pi'$ in any other way up to this point of the analysis.

\begin{assumption}\label{ass:PP'}
	$\pi(e_b) \leq \pi'(e'_b)$.
\end{assumption}

\begin{observation}\label{obs:querypathdecreasing}
	It holds that $\pi(e_1) < \pi(e_2) < \ldots < \pi(e_k)$ and $\pi'(e'_1) < \pi'(e'_2) < \ldots < \pi'(e'_k)$.
\end{observation}
\begin{proof}
	This holds because $\vec{P}$ is a query-path in $\pi$ and $\vec{P}'$ is a query-path in $\pi'$. Recall that query-paths correspond to recursions in the stack and Algorithm~\ref{alg:edgeoracle} only recurses on edges with lower rank than the current edge. Hence the edges along a query path must be decreasing in rank.
\end{proof}

\begin{claim}\label{cl:theoun12389}
	It holds that $\pi(e_b) = \pi'(e_{b+1})$. 
\end{claim}
\begin{proof}
	First, since $b+1 \geq 2$, we have $e_{b+1} \not= e_1$. Therefore by definition of $\phi(\pi, \vec{P})$, which rotates the ranks of $\pi$ in the reverse direction of $\vec{P}$ (i.e., in direction of $\cev{P}$), $\phi(\pi, \vec{P})(e_{b+1}) = \pi(e_{b})$. On the other hand, note that $e_{b+1} \not\in P'$ since $\cev{P}$ and $\cev{P'}$ branch after edge $e_b = e'_b$ and they are paths --- this implies that $\phi(\pi', \vec{P}')(e_{b+1}) = \pi'(e_{b+1})$. Combined with our assumption that $\phi(\pi, \vec{P}) = \phi(\pi', \vec{P}') = \sigma$, this implies the desired equality $\pi'(e_{b+1}) = \pi(e_{b})$.
\end{proof}

\begin{claim}\label{cl:dgcr128397}
	$\pi'(e_{b+1}) < \pi'(e'_b)$.
\end{claim}
\begin{proof}
	From Claim~\ref{cl:theoun12389} we know $\pi'(e_{b+1}) = \pi(e_b)$. Combined with Assumption~\ref{ass:PP'} that $\pi(e_b) \leq \pi'(e'_b)$, this implies $\pi'(e_{b+1}) \leq \pi'(e'_b)$. The equality can be ruled out since $\pi'$ is a permutation and distinct edges $e_{b+1}$ and $e'_b$ cannot be assigned the same rank. As such, $\pi'(e_{b+1}) < \pi'(e'_b)$.
\end{proof}

\begin{claim}\label{cl:pipi'equalonsmallranks}
	If $\pi(f) \not= \pi'(f)$ for some edge $f$, then $\pi(f) \geq \pi(e_b)$ and $\pi'(f) \geq \pi(e_b)$. In other words, the two permutations $\pi$ and $\pi'$ are identical on all ranks smaller than $\pi(e_b)$.
\end{claim}
\begin{proof}
	If $f \not \in P \cup P'$, then clearly $\pi'(f) = \pi(f) = \sigma(f)$ since $\phi(\pi, \vec{P}) = \phi(\pi', \vec{P}') = \sigma$ and $\phi$ only changes the ranks of edges along the query-path that it is given. So we can assume $f \in P \cup P'$; there are therefore four possible scenarios:
	
	\textbf{Case (1)} $f \in \{e_1, \ldots, e_{b-1}\}$: For any $i \in \{2, \ldots, b\}$, to have $\phi(\pi, \vec{P})(e_i) = \phi(\pi', \vec{P}')(e_i)$, it is necessary that $\pi(e_{i-1}) = \pi'(e_{i-1})$. Thus, in this case $\pi(f) = \pi'(f)$.
	
	\textbf{Case (2)} $f = e_b$: We have $\pi(f) = \pi(e_b)$ since $f = e_b$ and we have $\pi'(f) = \pi'(e_b) \geq \pi(e_b)$ by Assumption~\ref{ass:PP'}. So the claim holds in this case.	
	
	\textbf{Case (3)} $f = e_i$ for some $e_i \in \{e_{b+1}, \ldots, e_k\}$:  Observation~\ref{obs:querypathdecreasing} already implies $\pi(f) > \pi(e_b)$ in this case; it remains to prove $\pi'(f) \geq \pi(e_b)$. First, $\phi(\pi, \vec{P})(f) = \pi(e_{i-1}) \geq \pi(e_b)$ by construction of $\phi$ (as $i \not= 1$) and Observation~\ref{obs:querypathdecreasing} (as $i-1 \geq b$). Therefore, from $\phi(\pi, \vec{P}) = \phi(\pi', \vec{P'})$ (as assumed in Claim~\ref{cl:nobranch}) we get $\phi(\pi', \vec{P'})(f) \geq \pi(e_b)$. Given that $\phi(\pi', \vec{P'})(e) \leq \pi'(e)$ for every edge $e \not= e_1$ by construction of $\phi$ and monotonicity of $\pi'$ along $P'$ (Observation~\ref{obs:querypathdecreasing}) this means $\pi'(f) \geq \pi(e_b)$.
	
	\textbf{Case (4)} $f = e'_i$ for some $e'_i \in \{e'_{b+1}, \ldots, e'_k\}$: Oservation~\ref{obs:querypathdecreasing} and Assumption~\ref{ass:PP'} together imply $\pi'(f) > \pi'(e'_b) \geq \pi(e_b)$ in this case; it remains to prove $\pi(f) \geq \pi(e_b)$. First, $\phi(\pi', \vec{P'})(f) = \pi'(e'_{i-1}) \geq \pi'(e'_b)$ by construction of $\phi$ (as $i \not= 1$) and Observation~\ref{obs:querypathdecreasing} (as $i-1 \geq b$). Thus, from $\phi(\pi, \vec{P}) = \phi(\pi', \vec{P'})$ (as assumed in Claim~\ref{cl:nobranch}) we get $\phi(\pi, \vec{P})(f) \geq \pi'(e'_b)$. Given that $\phi(\pi, \vec{P})(e) \leq \pi(e)$ for every edge $e \not= e_1$ by construction of $\phi$ and monotonicity of $\pi$ along $P$ (Observation~\ref{obs:querypathdecreasing}) this means $\pi(f) \geq \pi'(e'_b) \geq \pi(e_b)$ where the last inequality follows from Assumption~\ref{ass:PP'}.
\end{proof}

\begin{claim}\label{cl:cclbu129837}
	$e_{b+1} \in \GMM{G, \pi'}$.
\end{claim}
\begin{proof}
	Assume for the sake of contradiction that $e_{b+1} \not\in \GMM{G, \pi'}$.  This means that $e_{b+1}$ must be incident to an edge $f \in \GMM{G, \pi'}$ with $\pi'(f) < \pi'(e_{b+1}) = \pi(e_b)$ (where recall the last equality follows from Claim~\ref{cl:theoun12389}). Since $\pi$ and $\pi'$ are identical for ranks smaller than $\pi(e_b)$ by  Claim~\ref{cl:pipi'equalonsmallranks}, $f \in \GMM{G, \pi'}$ implies $f \in \GMM{G, \pi}$ as well. Using this and combined with $\pi(f) = \pi'(f) < \pi(e_b)$, we show that $\vec{P}$ is not a valid query-path in permutation $\pi$ which is a contradiction. To see this, note that while running the edge oracle $\edgeoracle{e_{b+1}, \cdot, \pi}$, we should call $\edgeoracle{f, \cdot, \pi}$ before $\edgeoracle{e_b, \cdot, \pi}$ because $\pi(f) < \pi(e_b)$ and that $f$ and $e_b$ share the same endpoint with $e_{b+1}$; moreover, since $f \in \GMM{G, \pi}$, $\edgeoracle{e_{b+1}, \cdot, \pi}$ will immediately terminate and return $\false$ without calling $\edgeoracle{e_b, \cdot, \pi}$.
\end{proof}

We are now ready to finalize the proof of Claim~\ref{cl:nobranch} by showing that the assumptions above lead to a contradiction. The contradiction that we prove is that $\vec{P}'$ cannot be a valid query-path in permutation $\pi'$. To see this, recall that during the execution of $\edgeoracle{e'_{b+1}, \cdot, \pi'}$, we have to call $\edgeoracle{e_{b+1}, \cdot, \pi'}$ before $\edgeoracle{e'_b, \cdot, \pi'}$ since by Claim~\ref{cl:dgcr128397} $\pi'(e_{b+1}) < \pi'(e'_b)$. On the other hand, by Claim~\ref{cl:cclbu129837} the answer to $\edgeoracle{e_{b+1}, \pi'}$ is $\true$ and so $\edgeoracle{e'_{b+1}, \pi'}$ should terminate immediately and return $\false$ without calling $\edgeoracle{e'_b, \cdot, \pi'}$. Therefore, $\vec{P}'$ is not a valid query-path in $\pi'$ contradicting our assumption that $|\vec{P}| = |\vec{P}'|$ is possible. The proof of Claim~\ref{cl:nobranch} is thus complete. As discussed at the start of Section~\ref{sec:AL}, this also completes the proof of Lemma~\ref{lem:AL}.

\subsection{Proof of Lemma~\ref{lem:AU}}

To prove Lemma~\ref{lem:AU} we first recall a parallel implementation of the randomized greedy maximal matching algorithm and the bounds known for its {\em round-complexity}. For more details see \cite{BlellochSPAA, FischerTALG}.

\paragraph{Parallel Randomized Greedy MM:} Given a graph $G$ and a permutation $\pi$ over its edge-set, we repeat the following until $G$ becomes empty: in parallel add any ``local minimum'' edge $e$ to the matching and remove its endpoints from the graph. An edge is local minimum if its rank is smaller than that of all of its neighboring edges that remain in the graph.

It can be easily confirmed that the output of the algorithm above is exactly $\GMM{G, \pi}$. We use $\rho(G, \pi)$ to denote the round-complexity of the algorithm, i.e., the number of iterations that it takes until the graph becomes empty. 

It was shown in \cite{BlellochSPAA} that $\rho(G, \pi)$ can be bounded for a random permutation by $O(\log^2 n)$ with high probability. This was improved to $O(\log n)$ in \cite{FischerTALG}:

\begin{lemma}[{\cite{FischerTALG}}]\label{lem:parallel}
	Let $\pi$ be a permutation chosen uniformly at random over the edge-set of an $n$-vertex graph $G$. With probability at least $1-n^{-2}$, $\rho(G, \pi) = O(\log n)$.
\end{lemma}

We prove the following:

\begin{claim}\label{cl:queryvsround}
	Let $P$ be any query-path in $G$ for permutation $\pi$, then $\rho(G, \pi) \geq \lfloor \frac{|P|}{2} \rfloor$.
\end{claim}

Claim~\ref{cl:queryvsround} suffices to prove Lemma~\ref{lem:AU} as proved next.

\begin{proof}[Proof of Lemma~\ref{lem:AU} via Claim~\ref{cl:queryvsround}]
	For any permutation $\pi \in U$, by definition (\ref{eq:defLU}) there is at least one query-path of length at least $\beta + 1$ in permutation $\pi$. This implies by Claim~\ref{cl:queryvsround} that $\rho(G, \pi) \geq \lfloor \frac{\beta + 1}{2} \rfloor$ for any $\pi \in U$. If we set the constant $c$ in $\beta = c \log n$ to be sufficiently large, we can use Lemma~\ref{lem:parallel} to bound the probability of this event by $\leq 1/n^2$. Thus, $|U|/|\Pi| \leq 1/n^2$ which implies $|U| \leq |\Pi|/n^2 = m!/n^2$. The lemma follows noting that $|A_U| = |U|$ due to the one-to-one correspondence between vertices $A_U$ and permutations in $U$ that we discussed in Section~\ref{sec:AL}.
\end{proof}

\begin{proof}[Proof of Claim~\ref{cl:queryvsround}]
	Let $P = (e_k, \ldots, e_1)$ be our query-path and recall from Observation~\ref{obs:querypathdecreasing} that $\pi(e_k) > \ldots > \pi(e_1)$. For any edge $e$, we use $\rho(e)$ to denote the round of the parallel implementation in which edge $e$ gets removed from the graph according to permutation $\pi$. We show that for any $i \in \{2, \ldots, k-1\}$, $\rho(e_i) \geq \rho(e_{i-2}) + 1$. By a simple induction, this implies $\rho(e_{k-1}) \geq \lfloor k/2 \rfloor$ and so $\rho(G, \pi) \geq \lfloor k/2 \rfloor$.
	
	Suppose for the sake of contradiction that $\rho(e_i) < \rho(e_{i-2}) + 1$ (or equivalently $\rho(e_i) \leq \rho(e_{i-2})$) for some $2 \leq i \leq k-1$. This means that at the start of round $\rho(e_i)$ both $e_i$ and $e_{i-2}$ are still in the graph. Observe that if $\rho(e_{i-1}) < \rho(e_i)$, then during round $\rho(e_{i-1})$ at least one of the endpoints of $e_{i-1}$ must be matched and so either $e_i$ or $e_{i-2}$ (or both) should also be removed from the graph which contradicts $\rho(e_{i-1}) < \rho(e_i) \leq \rho(e_{i-2})$. Therefore $\rho(e_{i-1}) \geq \rho(e_i)$ and so $e_{i-1}$ is also still in the graph along with $e_{i}$ and $e_{i-2}$ at the start of round $\rho(e_i)$. This means that $e_i$ is not a local minimum during round $\rho(e_i)$ and so it is removed in this round because some other edge $f \not= e_{i}$ joins the matching. Note also that $f \not= e_{i-1}$ since $e_{i-1}$ is incident to $e_{i-2}$ and is not a local minimum. 
	
	Let $x$ be the vertex incident to both $e_{i-1}$ and $e_i$ and let $y$ be the other endpoint of $e_i$ incident to $e_{i+1}$ (note that since $i \leq k-1$ edge $e_{i+1}$ should exist). Since $f$ is a neighbor of $e_i$, it is either connected to $x$ or $y$. We show that both cases lead to contradictions. 
	
	\textbf{Case (1) } $f$ connected to $y$: Noting that $\pi(f) < \pi(e_i)$ since $\pi(f)$ is a local minimum when $e_i$ still is in the graph, we get that $\edgeoracle{e_{i+1}, \cdot, \pi}$ calls $\edgeoracle{f, \cdot, \pi}$ before $\edgeoracle{e_i, \cdot, \pi}$. But because $f \in \GMM{G, \pi}$, $\edgeoracle{e_{i+1}, \cdot, \pi}$ would not continue to call $\edgeoracle{e_i, \cdot, \pi}$ and $P$ cannot be a valid query-path.
	
	\textbf{Case (2)} $f$ connected to $x$: Since $f$ is a local minimum when $e_{i-1}$ still exists in the graph, $\pi(f) < \pi(e_{i-1})$. This implies that $\edgeoracle{e_i, \cdot, \pi}$ should call $\edgeoracle{f, \cdot, \pi}$ before $\edgeoracle{e_{i-1}, \cdot, \pi}$. But because $f \in \GMM{G, \pi}$, $\edgeoracle{e_i, \cdot, \pi}$ would not continue to call $\edgeoracle{e_{i-1}, \cdot, \pi}$ and so, again, $P$ cannot be a valid query-path.
	
	The contradictions above rule out the possibility of our assumption that $\rho(e_i) < \rho(e_{i-2}) + 1$ and so we indeed get $\rho(e_i) \geq \rho(e_{i-2}) + 1$ for all $2 \leq i \leq k-1$. As discussed, this implies $\rho(G, \pi) \geq \lfloor k/2 \rfloor$ completing the proof. 
\end{proof}

\begin{remark}\label{rem:MIS}
	Claim~\ref{cl:queryvsround} shows that for any permutation $\pi$, the maximum query length in the random greedy maximal matching algorithm is asymptotically upper bounded by the parallel round-complexity of this algorithm for the same permutation. One may wonder if this also holds for the randomized greedy maximal independent set (MIS) algorithm which processes the {\em vertices} in a random order and adds each encountered feasible vertex to the independent set greedily (see \cite{YoshidaSTOC09,BlellochSPAA}). Interestingly, the answer turns out to be negative. See Figure~\ref{fig:MIS}.
\end{remark}

\begin{figure}
\begin{tikzpicture}
    \node[anchor=south west,inner sep=0] at (0,0) {\includegraphics{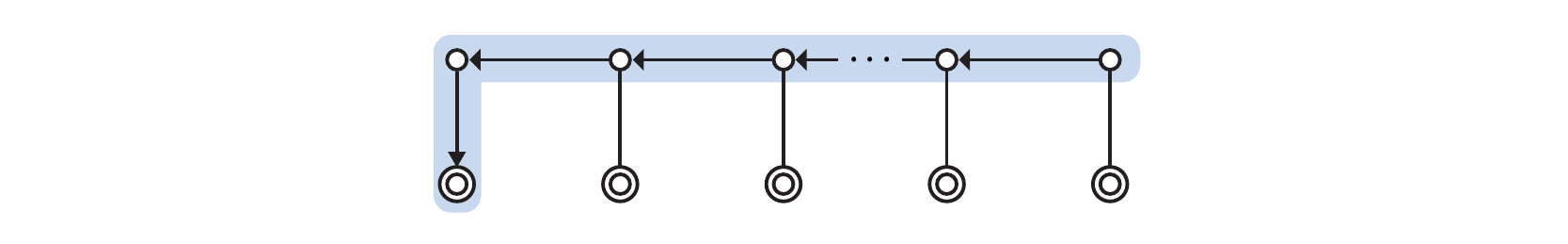}};
    \newcommand\XA{0.11}

    \node at (5.07-\XA,0.2) {$1$};
    \node at (5.07-\XA,2.5) {$2$};
    \node at (6.82-\XA,0.2) {$3$};
    \node at (6.82-\XA,2.5) {$4$};
    \node at (8.57-\XA,0.2) {$5$};
    \node at (8.58-\XA,2.5) {$6$};
    \node at (10.36-\XA,0.2) {$n-3$};
    \node at (10.36-\XA,2.5) {$n-2$};
    \node at (12.11-\XA,0.2) {$n-1$};
    \node at (12.11-\XA,2.5) {$n$};
    \node at (14.4-\XA,2) {(assume $n$ is even)};
\end{tikzpicture}
\caption{In this permutation, all the vertices with odd ranks join the MIS in round one and the whole graph becomes empty immediately, hence the parallel depth is 1. However, the query process for the vertex with rank $n$ first goes through all even nodes, thus has length $\geq n/2$.}
\label{fig:MIS}
\end{figure}
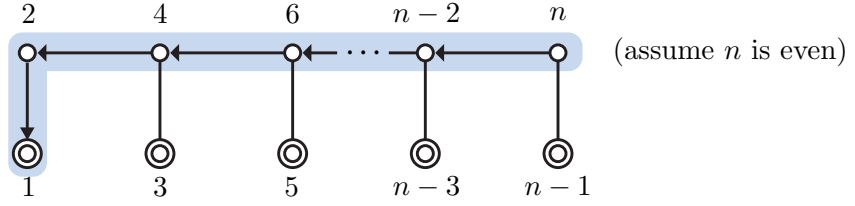

\section{The Final Algorithms for the Adjacency List Query Model}\label{sec:adjlist}

In this section, we show how the oracle analysis of Section~\ref{sec:query} can lead to our claimed bounds of Theorems~\ref{thm:adjlist-multiplicative} and \ref{thm:adjlist-additive} in the adjacency list query model.

As before, let $G=(V, E)$ be an arbitrary graph with $n$ vertices, $m$ edges, maximum degree $\Delta$, and average degree $\bar{d}$. Having defined and analyzed the oracle calls of the GMM algorithm in Section~\ref{sec:query}, we now employ the standard recipe of the literature in estimating the size of MCM or MVC. We sample a number of random vertices and simulate the vertex oracle of Algorithm~\ref{alg:vertexoracle} on each. If our sample size is sufficiently large, the fraction of matched sampled vertices is a good estimate of the fraction of vertices in the graph that are matched by GMM. To formalize this, we first show how to simulate the vertex and edge oracles of Section~\ref{sec:query} using adjacency list queries.

First, to generate the random permutation $\pi$, one can for each edge $e$ sample an independent rank $\sigma(e)$ which is a real in $[0, 1]$ chosen uniformly at random, and then obtain $\pi$ by sorting the edges in the increasing order of their ranks. This way we can expose the random permutation ``on the fly'' only where it is needed, avoiding the $\Omega(m)$ time needed for generating it for all the edges. Another challenge remains though. A trivial simulation of the edge oracle $\edgeoracle{e, \pi}$ (Algorithm~\ref{alg:edgeoracle}) is to generate the rank $\sigma(e')$ of all edges $e'$ incident to $e$ upon calling $\edgeoracle{e, \pi}$. The problem with this approach is that the total number of queries in answering $\vertexoracle{v, \pi}$ can be as large as $O(T(v, \pi) \Delta)$ whereas we need a bound of $\widetilde{O}(T(v, \pi))$ for our final results. Here $T(v, \pi)$ as defined in Section~\ref{sec:query} is the number of edges on which the edge oracle is recursively called during the execution of $\vertexoracle{v, \pi}$ and the $O(\Delta)$ factor comes from querying up to $O(\Delta)$ neighbors of each such edge. To get rid of this $\Delta$ factor, the idea is to expose the neighbors of each edge in ``batches,'' only when they are needed. This leads to the following bound, a variant of which was first proved by \cite{OnakSODA12}:

\begin{lemma}[\cite{OnakSODA12}]\label{lem:linear-implementation}
	Let $v$ be an arbitrary vertex in a graph $G=(V, E)$. There is an algorithm that draws a \underline{random} permutation $\pi$ over $E$, and  determines whether $v$ is matched in $\GMM{G, \pi}$ in time $\widetilde{O}(T(v, \pi) + 1)$ having query access to the adjacency lists. The algorithm succeeds w.h.p.
\end{lemma}

We note that Lemma~\ref{lem:linear-implementation} is slightly stronger than its variant proved in \cite{OnakSODA12} where the produced answers were only approximately close, in total variation distance, to the actual distribution. We observe that this can be turned to an exact guarantee by using the exact sublinear time binomial samplers of \cite{BinomialSampler,BinomialSampler2}. For completeness, we provide the full proof of Lemma~\ref{lem:linear-implementation} in Appendix~\ref{apx:implementation}. 

\subsection{Proof of Theorem~\ref{thm:adjlist-multiplicative}: Multiplicative Approximation}

We assume, w.l.o.g., in this section that the graph has no singleton vertices, and that the average degree $\bar{d}$ and maximum degree $\Delta$ are  given. Note that we can simply query the degree of every vertex in the graph, discard all the singleton vertices, and compute $\bar{d}$ and $\Delta$ for the rest of the vertices in $O(n)$ time and queries as allowed by Theorem~\ref{thm:adjlist-multiplicative}. Hence, the assumption comes w.l.o.g. Having this assumption, the rest of the algorithm of this section runs in $\widetilde{O}(\Delta/\epsilon^2)$ time.

As discussed, the general idea is to take $k$ random vertices and run the greedy oracle of Lemma~\ref{lem:linear-implementation} on them to see what fraction of them get matched. For the guarantee of Theorem~\ref{thm:adjlist-multiplicative}, it turns out that setting $k = \widetilde{\Theta}(\Delta / \bar{d} \epsilon^2)$ suffices to see sufficiently many matched vertices. The following simple claim plays a crucial role in arguing that these many samples suffice for our purpose: 

\begin{claim}\label{cl:matching-average-deg}
	For any $n$-vertex graph $G$ of maximum degree $\Delta$ and average degree $\bar{d}$, $\mu(G) \geq \frac{n \bar{d}}{4 \Delta}$.
\end{claim}
\begin{proof}
	By Vizing's theorem, any graph $G$ of maximum degree $\Delta$ has a proper $(\Delta+1)$-edge-coloring. Since the $m$ edges are colored only via $(\Delta+1)$ colors, there must be a color that is assigned to $\geq m/(\Delta+1)$ edges. Since the edges of any color form a matching, the graph must have a matching of size $\geq m/(\Delta + 1)$. Noting that $\bar{d} = 2m/n$, we get:
	$$
	\mu(G) \geq \frac{m}{\Delta + 1} = \frac{n \bar{d}/2}{\Delta + 1} \geq \frac{n \bar{d}}{4\Delta}.\qedhere
	$$
\end{proof}

Our starting point is the following Algorithm~\ref{alg:adjlist-multiplicative}:

\begin{algorithm}[H]\label{alg:adjlist-multiplicative}
\caption{An algorithm used for Theorem~\ref{thm:adjlist-multiplicative}, given parameter $\epsilon > 0$.}	
	$k \gets 128 \cdot 24 (\Delta \ln n) /(\epsilon^2 \bar{d}).$ \tcp*{Note that $\bar{d}$ and $\Delta$ are known to the algorithm.}
	
	Sample $k$ vertices $v_1, \ldots, v_k$ (with replacement) independently and uniformly from $V$.
	
	For each $i \in [k]$ run the algorithm of Lemma~\ref{lem:linear-implementation} on vertex $v_i$. For each $i \in [k]$ let $X_i$ be the indicator of the event that $v_i$ is matched once we run Lemma~\ref{lem:linear-implementation} for $v_i$.
	
	Let $X \gets \sum_{i=1}^k X_i$ and let $f \gets X/k$ be the fraction of vertices $v_1, \ldots, v_k$ that get matched.
	
	Let $\widetilde{\mu} \gets (1-\frac{\epsilon}{2})f n / 2$ and let $\widetilde{\nu} \gets (1+\frac{\epsilon}{2})f n$.
	
	\Return $\widetilde{\mu}$ as the estimate for $\mu(G)$ and \Return $\widetilde{\nu}$ as the estimate for $\nu(G)$.
\end{algorithm}

We start by analyzing the approximation ratio of Algorithm~\ref{alg:adjlist-multiplicative}:

\begin{lemma}\label{lem:adjlist-multiplicative-approx}
	Let $\widetilde{\mu}$ and $\widetilde{\nu}$ be the outputs of Algorithm~\ref{alg:adjlist-multiplicative}. With probability $1-2	n^{-4}$,
	$$
	(1-\epsilon) \frac{1}{2} \mu(G) \leq \widetilde{\mu} \leq \mu(G)
	\qquad\&\qquad
	\nu(G) \leq  \widetilde{\nu} \leq (1+\epsilon)2 \nu(G).
	$$
\end{lemma}

\begin{proof}
Let us now measure the expected value of our estimates. From our definition, $X_i = 1$ if and only if a random vertex $v_i$ for a random permutation $\pi$ is matched in $\GMM{G, \pi}$. Since the number of vertices matched in a matching is twice the size of the matching, this implies that
$$
	\E[X_i] = \Pr_{v_i, \pi}[X_i = 1] = \frac{2\E_\pi|\GMM{G, \pi}|}{n}.
$$
As such,
\begin{equation}\label{eq:cllrcg128937}
	\E[X] = \E[X_1 + \ldots + X_k] = \frac{2 k\E_\pi|\GMM{G, \pi}|}{n}.
\end{equation}
Since $X$ is sum of independent Bernoulli random variables, by the Chernoff bound (Proposition~\ref{prop:chernoff}):
\begin{equation}\label{eq:hcll92123}
	\Pr\left[|X - \E[X]| \geq \sqrt{12 \E[X] \ln n}\right] \leq 2 \exp\left(- \frac{12 \E[X] \ln n}{3 \E[X]}\right) = 2/n^{4}.
\end{equation}
Noting from Algorithm~\ref{alg:adjlist-multiplicative} that $f \cdot n = Xn/k$, inequality (\ref{eq:hcll92123}) implies that with probability $1-2/n^4$, 
\begin{flalign*}
	f \cdot n &\in \frac{(\E[X] \pm \sqrt{12 \E[X] \ln n})n}{k}\\
	&= \frac{\E[X] n}{k} \pm \sqrt{12 \E[X] n^2 k^{-2} \ln n} \\
	&= 2\E_\pi|\GMM{G, \pi}| \pm  \sqrt{24 \E_\pi|\GMM{G, \pi}| n k^{-1} \ln n} \tag{By (\ref{eq:cllrcg128937}).}\\
	&= 2\E_\pi|\GMM{G, \pi}| \pm  \sqrt{\frac{\E_\pi|\GMM{G, \pi}| \epsilon^2 n \bar{d}}{128 \Delta}}. \tag{Since $k = 128 \cdot 24 \frac{\Delta \ln n}{\epsilon^2 \bar{d}}$.}
\end{flalign*}
Note that $\mu(G) \leq 2\E_\pi|\GMM{G, \pi}|$ and also recall from Claim~\ref{cl:matching-average-deg} that $\mu(G) \geq \frac{n \bar{d}}{4 \Delta}$. As such, we have $\frac{n \bar{d}}{4 \Delta} \leq 2\E_\pi|\GMM{G, \pi}|$. Combined with the range above, with probability $1-2/n^4$ we have
$$
f \cdot n \in 2 \E_\pi|\GMM{G, \pi}| \pm \sqrt{\frac{(\E_\pi|\GMM{G, \pi}|)^2 \epsilon^2}{16}} = \left(2 \pm \frac{\epsilon}{4}\right) \E_\pi|\GMM{G, \pi}|.
$$
Since $\widetilde{\mu} = (1-\frac{\epsilon}{2})f \cdot n/2$ and $\widetilde{\nu} = (1+\frac{\epsilon}{2}) f \cdot n$, this means
\begin{flalign}
	(1-\epsilon) \E_\pi|\GMM{G, \pi}| \leq\,\, &\widetilde{\mu} \, \leq \E_\pi|\GMM{G, \pi}|,\label{eq:lntb129387}\\
	2\E_\pi|\GMM{G, \pi}| \leq \,\, &\widetilde{\nu} \, \leq (1+\epsilon) 2\E_\pi|\GMM{G, \pi}|\label{eq:lntb129387-2}.
\end{flalign}
Next, observe that $\frac{1}{2} \mu(G) \leq \E_\pi|\GMM{G, \pi}| \leq \mu(G)$ since a maximal matching is a $2$-approximate maximum matching, and $\nu(G) \leq 2 \E_\pi|\GMM{G, \pi}| \leq 2 \nu(G)$ since the set of vertices of a maximal matching is a $2$-approximate minimum vertex cover. These, plugged into (\ref{eq:lntb129387}) and (\ref{eq:lntb129387-2}) give the desired inequalities of the lemma, completing the proof.
\end{proof}

We are now ready to prove Theorem~\ref{thm:adjlist-multiplicative}.

\begin{proof}[Proof of Theorem~\ref{thm:adjlist-multiplicative}]
	By Theorem~\ref{thm:querycomplexity}, for a random vertex $v \in V$, $\E_{v, \pi}[T(v, \pi)] = O(\bar{d} \cdot \log n)$. As such, for each vertex $v_i$ of Algorithm~\ref{alg:adjlist-multiplicative}, the algorithm of Lemma~\ref{lem:linear-implementation} takes $\widetilde{O}(\bar{d} + 1)$ expected time to determine whether it is matched in a random permutation. Since we run this algorithm for $k = \widetilde{O}(\Delta / \epsilon^2 \bar{d})$ vertices, the \underline{expected} running time of Algorithm~\ref{alg:adjlist-multiplicative} is $\widetilde{O}(\Delta / \epsilon^2 \bar{d}) \cdot \widetilde{O}(\bar{d} + 1) = \widetilde{O}(\Delta / \epsilon^2)$ where the latter equality crucially uses our assumption of the start of the section that there are no singleton vertices in the graph, which implies $1/\bar{d} = O(1)$.	
	
	To achieve the high probability bound on the time-complexity, we run $\Theta(\log n)$ instances of Algorithm~\ref{alg:adjlist-multiplicative} in parallel and return the output of the instance that terminates first. By Markov's inequality, each instance terminates in time $\widetilde{O}(\Delta/\epsilon^2)$ with a constant probability. As such, at least one of the instances terminates in $\widetilde{O}(\Delta/\epsilon^2)$ time with probability $1-1/\poly(n)$. Recall also that we spent $O(n)$ time at the start of the section to throw away singleton vertices and compute $\bar{d}$ and $\Delta$. As such, the total time complexity is, w.h.p., $O(n) + \widetilde{O}(\Delta/\epsilon^2)$. On the other hand, since the approximation ratio guarantee of Lemma~\ref{lem:adjlist-multiplicative-approx} holds with probability $1-1/\poly(n)$ for each instance, {\em all} $O(\log n)$ instances (including the one that terminates first) achieve the claimed approximation with a high probability of $1-1/\poly(n)$.
\end{proof}

\subsection{Proof of Theorem~\ref{thm:adjlist-additive}: Multiplicative-Additive Approximation}

The algorithm we use for Theorem~\ref{thm:adjlist-additive} is similar to Algorithm~\ref{alg:adjlist-multiplicative} for Theorem~\ref{thm:adjlist-multiplicative} except that the additive $\epsilon n$ error of Theorem~\ref{thm:adjlist-additive} allows us to take $\widetilde{O}(1/\epsilon^2)$ sample vertices instead of $\widetilde{O}(\Delta/\epsilon^2\bar{d})$ as in Algorithm~\ref{alg:adjlist-multiplicative}. Formally, we use the following Algorithm~\ref{alg:adjlist-additive}:

\begin{algorithm}[H]\label{alg:adjlist-additive}
\caption{An algorithm used for Theorem~\ref{thm:adjlist-additive}, given parameter $\epsilon > 0$.}	
	$k \gets 16 \cdot 24 \ln n/ \epsilon^2.$
	
	Sample $k$ vertices $v_1, \ldots, v_k$ (with replacement) independently and uniformly from $V$.
	
	For each $i \in [k]$ run the algorithm of Lemma~\ref{lem:linear-implementation} on vertex $v_i$. For each $i \in [k]$ let $X_i$ be the indicator of the event that $v_i$ is matched once we run Lemma~\ref{lem:linear-implementation} for $v_i$.
	
	Let $X \gets \sum_{i=1}^k X_i$ and let $f \gets X/k$ be the fraction of vertices $v_1, \ldots, v_k$ that get matched.
	
	Let $\widetilde{\mu} \gets \frac{f n}{2} - \frac{\epsilon}{2} n$ and let $\widetilde{\nu} \gets f n + \frac{\epsilon}{4} n$.
	
	\Return $\widetilde{\mu}$ as the estimate for $\mu(G)$ and \Return $\widetilde{\nu}$ as the estimate for $\nu(G)$.
\end{algorithm}

\begin{lemma}\label{lem:adjlist-additive-approx}
	Let $\widetilde{\mu}$ and $\widetilde{\nu}$ be the outputs of Algorithm~\ref{alg:adjlist-additive}. With probability $1-2	n^{-4}$,
	$$
	\frac{1}{2} \mu(G) - \epsilon n \leq \widetilde{\mu} \leq \mu(G)
	\qquad\&\qquad
	\nu(G) \leq  \widetilde{\nu} \leq 2 \nu(G) + \epsilon n.
	$$
\end{lemma}
\begin{proof}
	Observe that inequalities (\ref{eq:cllrcg128937}) and (\ref{eq:hcll92123}) that we proved for Algorithm~\ref{alg:adjlist-multiplicative}  hold for Algorithm~\ref{alg:adjlist-additive} too for exactly the same reasons. Particularly, inequality (\ref{eq:hcll92123}) implies that with probability $1-2/n^4$, 
\begin{flalign*}
	f \cdot n &\in \frac{(\E[X] \pm \sqrt{12 \E[X] \ln n})n}{k}\\
	&= \frac{\E[X] n}{k} \pm \sqrt{12 \E[X] n^2 k^{-2} \ln n} \\
	&= 2\E_\pi|\GMM{G, \pi}| \pm  \sqrt{24 \E_\pi|\GMM{G, \pi}| n k^{-1} \ln n} \tag{By (\ref{eq:cllrcg128937}).}\\
	&= 2\E_\pi|\GMM{G, \pi}| \pm  \sqrt{\E_\pi|\GMM{G, \pi}| \epsilon^2 n/16}. \tag{Since $k = 16 \cdot 24 \frac{\ln n}{\epsilon^2}$.}\\
	&\in 2\E_\pi|\GMM{G, \pi}| \pm  \epsilon n /4. \tag{Since $\E_\pi|\GMM{G, \pi}| \leq n$.}
\end{flalign*}
Combined with $\widetilde{\mu} = \frac{f n}{2} - \frac{\epsilon}{2} n$ and $\widetilde{\nu} = f n + \frac{\epsilon }{4} n$, this implies that, w.h.p., 
\begin{flalign}
	\E_\pi|\GMM{G, \pi}| - \epsilon n \leq\,\, &\widetilde{\mu} \, \leq \E_\pi|\GMM{G, \pi}|,\label{eq:bbcclltt}\\
	2\E_\pi|\GMM{G, \pi}| \leq \,\, &\widetilde{\nu} \, \leq 2\E_\pi|\GMM{G, \pi}|\label{eq:bbcclltt-2} + \epsilon n.
\end{flalign}
Plugging $\frac{1}{2} \mu(G) \leq \E_\pi|\GMM{G, \pi}| \leq \mu(G)$ and $\nu(G) \leq 2 \E_\pi|\GMM{G, \pi}| \leq 2 \nu(G)$ into (\ref{eq:bbcclltt}) and (\ref{eq:bbcclltt-2}) completes the proof.
\end{proof}

\begin{proof}[Proof of Theorem~\ref{thm:adjlist-additive}]
	As discussed in the proof of Theorem~\ref{thm:adjlist-multiplicative}, each call to Lemma~\ref{lem:linear-implementation} for a randomly chosen vertex takes $\widetilde{O}(\bar{d} + 1)$ expected time. Since $k = \widetilde{O}(1/\epsilon^2)$, Algorithm~\ref{alg:adjlist-additive} takes $\widetilde{O}((\bar{d} + 1)/\epsilon^2)$ expected time in total.
	
	To achieve the high probability bound on the time-complexity, as in the proof of Theorem~\ref{thm:adjlist-multiplicative}, we run $\Theta(\log n)$ instances of Algorithm~\ref{alg:adjlist-additive} in parallel and return the output of the instance that terminates first. By Markov's inequality, each instance terminates in time $\widetilde{O}((\bar{d} + 1)/\epsilon^2)$ with a constant probability. As such, at least one of the instances terminates in $\widetilde{O}((\bar{d} + 1)/\epsilon^2)$ time with probability $1-1/\poly(n)$. On the other hand, since the approximation guarantee of Lemma~\ref{lem:adjlist-additive-approx} holds with probability $1-1/\poly(n)$ for each instance, {\em all} $O(\log n)$ instances (including the one that terminates first) achieve a $(2, \epsilon n)$-approximation with  probability $1-1/\poly(n)$.
\end{proof}

\section{The Final Algorithm for the Adjacency Matrix Query Model}\label{sec:adjmatrix}
In this section, we prove Theorem~\ref{thm:adjmatrix}. The first challenge is that Lemma~\ref{lem:linear-implementation} is based on adjacency list queries. As such, we need a way of implementing these adjacency list queries in the adjacency matrix model. To do this, we transform the input graph $G$ to another graph $H$ where the adjacency list queries to $H$ can be implemented efficiently via adjacency matrix queries to $G$ and at the same time the oracle calls on $H$ suffice to approximate the size of MCM and MVC for $G$. 

We note that another such reduction was given before by \cite{OnakSODA12}. However, the reduction of \cite{OnakSODA12} is not applicable in our case since it, crucially, adds parallel edges and self-loops to $H$ while Theorem~\ref{thm:querycomplexity} is proved for simple graphs.

\def\x{3.85}
\begin{figure}[t]
\centering
\includegraphics[scale=0.9]{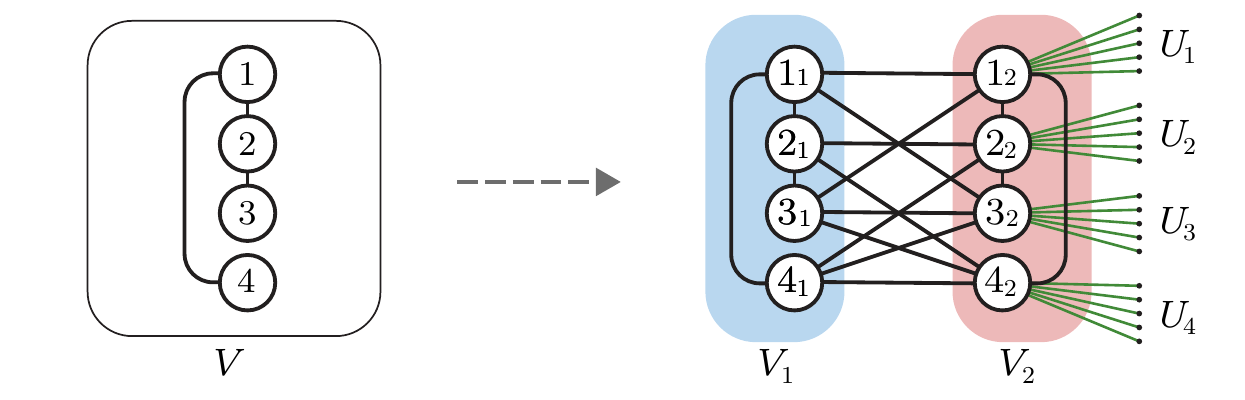}
\caption{An example of our reduction.} 
\label{fig:matrix}
\end{figure}

We first make a mild assumption that $\epsilon$ in Theorem~\ref{thm:adjmatrix} is at least $1/n$, noting that otherwise $\widetilde{O}(n/\epsilon^3)$ is large enough to query the whole graph, making Theorem~\ref{thm:adjmatrix} trivial.

Let us now formalize the construction of graph $H=(V_H, E_H)$ from the input graph $G=(V, E)$ (see Figure~\ref{fig:matrix} for an illustration). Throughout this section, we continue to use $n$ to denote the number of vertices in the original graph $G$. Letting $s := 10 n/\epsilon$, the construction is as follows:
\begin{itemize}
	\item As for the vertex-set $V_H$ we have $V_H = V_1 \cup V_2 \cup U_1 \cup \ldots \cup U_n$ where:
	\begin{itemize}
		\item $V_1 := \{1_1, 2_1, \ldots, n_1\}$ and $V_2 := \{1_2, 2_2, \ldots, n_2\}$ are two copies of $V$ both with size $n$.
		\item For each $i \in [n]$, subset $U_i = \{1'_i, 2'_i, \ldots, s'_i \}$ has size $s$.
	\end{itemize}
	Note in particular that $H$ has $2n + ns = \Theta(n^2/\epsilon)$ vertices.
	\item We define the edge-set $E_H$ by specifying the adjacency lists of the vertices in $V_H$ while being careful that each edge $(u, v) \in E_H$ appears in the adjacency lists of both $u$ and $v$. See Figure~\ref{fig:matrix} for reference.
	\begin{itemize}
		\item Let $v \in [n]$. Vertex $v_1 \in V_1$ has degree exactly $n$ in $H$. For any $i \in [n]$, if $(v, i) \in E$ then the $i$-th neighbor of $v_1$ is vertex $i_1 \in V_1$, otherwise it is vertex $i_2 \in V_2$.
		\item Let $v \in [n]$. Vertex $v_2 \in V_2$ has degree exactly $n + s$. For any $i \in [n]$, if $(v, i) \in E$ then the $i$-th neighbor of $v_2$ is vertex $i_2 \in V_2$, otherwise it is vertex $i_1 \in V_1$. The last $s$ vertices in the adjacency list of $v_2$ are the vertices in $U_{v}$.
		\item Each vertex $u \in U_v$ for any $v \in [n]$ has exactly one neighbor $v_2 \in V_2$.
	\end{itemize}
	Note that $H[V_1]$ and $H[V_2]$ are both isomorphic to $G$.
\end{itemize}

The following observation follows immediately from the construction above and the way neighbors of each vertex are ordered in their adjacency lists:

\begin{observation}\label{obs:matrix-adjlistqueries}
	For any vertex $v \in V_H$, $\deg_{H}(v)$ does not depend on the edges in $G$ and is, therefore, known a priori with zero queries to $G$. Furthermore, for any $v \in V_H$ and any $i \in [\deg_{H}(v)]$ one can determine the $i$-th edge of $v$ in $H$ by making at most one adjacency matrix query to $G$.
\end{observation}

Let us use $T_{H}(v, \pi)$ instead of $T(v, \pi)$ (defined in Section~\ref{sec:query} and used in Lemma~\ref{lem:linear-implementation}) to emphasize that we run RGMM on graph $H$ and not $G$ in this section.

\begin{claim}\label{cl:query-size-from-V1}
Let $\pi$ be a random permutation over the edge-set $E_H$ of $H$. For a vertex $v$ chosen uniformly at random from $V_1$ and independently from $\pi$,
$$
\E_{v \sim V_1, \pi}[T_{H}(v, \pi)] = \widetilde{O}(n/\epsilon).
$$
\end{claim}
\begin{proof}
	By the construction, graph $H$ has $|E_H| = O(n^2 + ns) = \widetilde{\Theta}(n^2/\epsilon)$ edges and $|V_H| = 2n + n s = \Theta(n^2/\epsilon)$ vertices. Applying Theoroem~\ref{thm:querycomplexity}, for a vertex $v \in V_H$ chosen uniformly at random, we get 
	$$
	\E_{v \sim V_H, \pi}[T_{H}(v, \pi)] = O\left(\frac{|E_H|}{|V_H|} \log |V_H|\right).
	$$
	Instead of querying a random vertex $v \in V_H$, suppose that we sum over all of them. This gives:
	\begin{equation}\label{eq:llrcggjjht198273}
	\sum_{v \in V_H} \E_\pi[T_{H}(v, \pi)] = |V_H| \cdot \E_{v \sim V_H, \pi}[T_{H}(v, \pi)] = O\left(|E_H| \log |V_H|\right)  = \widetilde{O}(n^2/\epsilon).
	\end{equation}
	Finally, since $|V_1| = n$, for a vertex $v$ chosen uniformly at random from $V_1$, we have
	$$
	\E_{v \sim V_1, \pi}[T_{H}(v, \pi)] \leq \left(\sum_{v \in V_H} \E_\pi[T_{H}(v, \pi)]\right)/|V_1| \stackrel{(\ref{eq:llrcggjjht198273})}{=} \widetilde{O}(n^2/\epsilon) / n = \widetilde{O}(n/\epsilon).\qedhere
	$$
\end{proof}

Now that we know queries starting from a random vertex in $V_1$ can be implemented efficiently in $\widetilde{O}(n/\epsilon)$ expected time, the question is how can we use these queries to get our estimators for the size of maximum matching and minimum vertex cover of the original graph $G$. 

For any permutation $\pi$ over $E_H$, let us define
$$
M_1(\pi) := \GMM{H, \pi} \cap (V_1 \times V_1)
$$
to be the set of edges in matching $\GMM{H, \pi}$ with both endpoints in $V_1$. The following Claim~\ref{cl:M-is-a-good-estimator} shows that the expected value of $M_1(\pi)$ for a random permutation $\pi$ can be used for approximating both maximum matching and the minimum vertex cover of $G$. 

The idea behind Claim~\ref{cl:M-is-a-good-estimator} is as follows. Note that by construction, $H[V_1]$ is isomorphic to $G$ and so for any $\pi$, there is a matching in $G$ that corresponds to $M_1(\pi)$. However, the vertices in $V_1$ can also be matched to the vertices in $V_2$ in graph $H$, and so $M_1(\pi)$ is not necessarily a {\em maximal} matching of $H[V_1]$. The key insight, however, is that the vast majority of vertices $v_2 \in V_2$ will actually be matched to their neighbors in $U_v$ in a RGMM of $H$. Hence, $M_1(\pi)$ for a random permutation $\pi$ is indeed close to a maximal matching of $H[V_1]$ and its size can be used to approximate both the size of MCM and that of the MVC of $G$. Formally:

\begin{claim}\label{cl:M-is-a-good-estimator}
	It holds that
	$$
	\frac{1}{2} \mu(G) - \frac{\epsilon}{20} n \leq \E_\pi|M_1(\pi)| \leq \mu(G) \qquad \& \qquad \nu(G) - \frac{\epsilon}{10} n \leq 2\E_\pi|M_1(\pi)| \leq 2\nu(G).
	$$
\end{claim}

\begin{proof}
	Let us use $B_1(\pi)$ to denote the set of vertices in $V_1$ that are matched to $V_2$ in $\GMM{H, \pi}$. We first show that for every permutation $\pi$ over $E_H$, it holds that
	\begin{equation}\label{eq:llcgggbmbwx123897}
	\frac{1}{2} (\mu(G) - |B_1(\pi)|) \leq |M_1(\pi)| \leq \mu(G) \qquad \& \qquad \nu(G) - |B_1(\pi)| \leq 2|M_1(\pi)| \leq 2\nu(G).
	\end{equation}
	
	For the first inequality, observe that $M_1(\pi)$ is a matching of $H[V_1]$ and $H[V_1]$ is isomorphic to $G$ by the construction; thus, clearly $|M_1(\pi)| \leq \mu(G)$.
	
	For the second inequality, observe that $M_1(\pi)$ is a maximal matching of $H[V_1 \setminus B_1(\pi)]$ which is isomorphic to $G[V \setminus B_1(\pi)]$. Since a maximal matching is at least half the size of a maximum matching, we get $|M_1(\pi)| \geq \frac{1}{2} \mu(G[V \setminus B_1(\pi)]) \geq \frac{1}{2} (\mu(G) - |B_1(\pi)|)$.
	
	For the third inequality, note that since $M_1(\pi)$ is a maximal matching of $H[V_1 \setminus B_1(\pi)]$, the set of vertices matched by it covers all edges in $H[V_1 \setminus B_1(\pi)]$. Adding the vertices in $B_1(\pi)$ to this set, we cover all edges in $H[V_1]$. Hence $\nu(H[V_1]) \leq 2|M_1(\pi)|+|B_1(\pi)|$. The inequality then follows from $H[V_1]$ being isomorphic to $G$.
	
	For the fourth inequality, note that since $M_1(\pi)$ is a matching in $H[V_1]$ and each vertex can cover at most one edge of it, we have $|M_1(\pi)| \leq \nu(H[V_1])$. Given that $H[V_1]$ is isomorphic to $G$, we thus get $|M_1(\pi)| \leq \nu(G)$. Multiplying through by a factor of 2 and adding $|B_1(\pi)|$ to both sides, we get $2|M_1(\pi)| + |B_1(\pi)| \leq 2\nu(G) + |B_1(\pi)|$.

	Now, observe that in any permutation $\pi$, if the lowest rank neighbor of a vertex $v_2 \in V_2$ is connected to $U_v$, then this edge is in matching $\GMM{\pi}$ since the vertices in $U_v$ have degree exactly one. Moreover, since each vertex $v_2 \in V_2$ has degree exactly $n+ s$ and $s$ of these neighbors are to $U_v$, the minimum rank edge of $v_2$ is indeed connected to $U_v$ with probability $\frac{s}{n + s} = \frac{10n/\epsilon}{n + 10n/\epsilon} \geq 1-\epsilon/10$. As such the number of vertices in $V_2$ that are matched to $V_1$ is at most $\frac{\epsilon}{10}|V_2| = \frac{\epsilon}{10}n$ in expectation taken over $\pi$. This, in turn, implies that $\E_\pi|B_1(\pi)| \leq \frac{\epsilon}{10}n$. Taking expectation over a random $\pi$ in (\ref{eq:llcgggbmbwx123897}) and plugging this upper bound for $\E_\pi|B_1(\pi)|$ proves the claim.
\end{proof}

Finally, the algorithm we use for Theorem~\ref{thm:adjmatrix} is formalized below as Algorithm~\ref{alg:adjmatrix}.

\begin{algorithm}
	Let $H=(V_H, E_H)$ be the graph described above (we do not explicitly construct $H$ here).
		
	$k \gets 16 \cdot 24 \ln n/ \epsilon^2.$
	
	Sample $k$ vertices $v^1, \ldots, v^k$ (with replacement) independently and uniformly from $V_1$.
	
	For each $i \in [k]$ run the algorithm of Lemma~\ref{lem:linear-implementation} on vertex $v^i$, for graph $H$. Let $X_i$ be the indicator of the event that $v^i$ is matched \underline{to another vertex in $V_1$}.
	
	Let $X \gets \sum_{i=1}^k X_i$ and let $f \gets X/k$ be the fraction of $v_1, \ldots, v_k$ that get matched to $V_1$.

	Let $\widetilde{\mu} \gets \frac{f n}{2} - \frac{\epsilon}{2} n$ and $\widetilde{\nu} \gets f n + \frac{\epsilon}{2} n$.
	
	\Return $\widetilde{\mu}$ as the estimate for $\mu(G)$ and \Return $\widetilde{\nu}$ as the estimate for $\nu(G)$.

	\caption{An algorithm used for Theorem~\ref{thm:adjmatrix}, given parameter $\epsilon > 0$.}
	\label{alg:adjmatrix} 
\end{algorithm}

Let us analyze the approximation ratio of Algorithm~\ref{alg:adjmatrix}.

\begin{lemma}\label{lem:adjmatrix-approx}
	Let $\widetilde{\mu}$ and $\widetilde{\nu}$ be the outputs of Algorithm~\ref{alg:adjmatrix}. With probability $1-2	n^{-4}$,
	$$
	\frac{1}{2} \mu(G) - \epsilon n \leq \widetilde{\mu} \leq \mu(G)
	\qquad\&\qquad
	\nu(G) \leq  \widetilde{\nu} \leq 2 \nu(G) + \epsilon n.
	$$
\end{lemma}

\begin{proof}
Let us now measure the expected value of our estimates. From our definition, $X_i = 1$ if and only if a random vertex $v^i \in V_1$ for a random permutation $\pi$ is matched to another vertex of $V_1$ in $\GMM{H, \pi}$. Recall that we defined $M_1(\pi)$ to be the set of edges in $\GMM{H, \pi} \cap (V_1 \times V_1)$. Combined with the fact that the number of vertices matched in a matching is twice the size of the matching, this implies that
$$
	\E[X_i] = \Pr_{v^i, \pi}[X_i = 1] = \frac{2\E_\pi|M_1(\pi)|}{|V_1|} = \frac{2\E_\pi|M_1(\pi)|}{n}.
$$
As such,
\begin{equation}\label{eq:ajcllrcg128937}
	\E[X] = \E[X_1 + \ldots + X_k] = \frac{2 k\E_\pi|M_1(\pi)|}{n}.
\end{equation}
Since $X$ is sum of independent Bernoulli random variables, by the Chernoff bound (Proposition~\ref{prop:chernoff}):
\begin{equation}\label{eq:ajhcll92123}
	\Pr\left[|X - \E[X]| \geq \sqrt{12 \E[X] \ln n}\right] \leq 2 \exp\left(- \frac{12 \E[X] \ln n}{3 \E[X]}\right) = 2/n^{4}.
\end{equation}
Noting from Algorithm~\ref{alg:adjmatrix} that $f \cdot n = Xn/k$, inequality (\ref{eq:ajhcll92123}) implies that with probability $1-2/n^4$, 
\begin{flalign*}
	f \cdot n &\in \frac{(\E[X] \pm \sqrt{12 \E[X] \ln n})n}{k}\\
	&= \frac{\E[X] n}{k} \pm \sqrt{12 \E[X] n^2 k^{-2} \ln n} \\
	&= 2\E_\pi|M_1(\pi)| \pm  \sqrt{24 \E_\pi|M_1(\pi)| n k^{-1} \ln n} \tag{By (\ref{eq:ajcllrcg128937}).}\\
	&= 2\E_\pi|M_1(\pi)| \pm  \sqrt{\E_\pi|M_1(\pi)| \epsilon^2 n/16} \tag{Since $k = 16 \cdot 24 \frac{\ln n}{\epsilon^2}$.}\\
	&\in 2\E_\pi|M_1(\pi)| \pm  \frac{\epsilon n}{4}.\tag{Since $\E_\pi|M_1(\pi)| \leq |V_1| = n$.}
\end{flalign*}

Combined with $\frac{1}{2} \mu(G) - \frac{\epsilon}{20} n \leq \E_\pi|M_1(\pi)| \leq \mu(G)$ from Claim~\ref{cl:M-is-a-good-estimator}, and $\widetilde{\mu} = \frac{f n}{2} - \frac{\epsilon}{2} n$, the range above for $fn$ implies $\frac{1}{2} \mu(G) - \epsilon n < \widetilde{\mu} < \mu(G)$. 

Combined with $\nu(G) - \frac{\epsilon}{10} n \leq 2\E_\pi|M_1(\pi)| \leq 2\nu(G)$ from Claim~\ref{cl:M-is-a-good-estimator} and $\widetilde{\nu} = f n + \frac{\epsilon}{2}n$, the range above for $fn$ implies $\nu(G) < \widetilde{\nu} < 2 \nu(G) + \epsilon n$.
\end{proof}

\begin{proof}[Proof of Theorem~\ref{thm:adjmatrix}]
	By Claim~\ref{cl:query-size-from-V1} each call to Lemma~\ref{lem:linear-implementation} for a random vertex from $V_1$ leads to $\widetilde{O}(n/\epsilon)$ adjacency list queries to $H$. By Observation~\ref{obs:matrix-adjlistqueries}, each of these adjacency list queries to $H$, can be implement with one adjacency matrix query to $G$ in $O(1)$ time. The total expected number of adjacency matrix queries to $G$ and the time-complexity of the algorithm, is therefore, $k \cdot \widetilde{O}(n/\epsilon) = \widetilde{O}(n/\epsilon^3)$.
	
	To achieve the high probability bound on the time-complexity, as in the proofs of Theorems~\ref{thm:adjlist-multiplicative} and \ref{thm:adjlist-additive}, we run $\Theta(\log n)$ instances of Algorithm~\ref{alg:adjmatrix} in parallel and return the output of the instance that terminates first. By Markov's inequality, each instance terminates in time $\widetilde{O}(n/\epsilon^3)$ with a constant probability. As such, at least one of the instances terminates in $\widetilde{O}(n/\epsilon^3)$ time with probability $1-1/\poly(n)$. On the other hand, since the approximation guarantee of Lemma~\ref{lem:adjmatrix-approx} holds with probability $1-1/\poly(n)$ for each instance, {\em all} $O(\log n)$ instances (including the one that terminates first) achieve a $(2, \epsilon n)$-approximation with probability $1-1/\poly(n)$.
\end{proof}


\section{Acknowledgements}

I thank Sanjeev Khanna for many helpful discussions, and additionally thank Yu Chen, Sanjeev Khanna, and Zihan Tan for insightful comments. 

\bibliographystyle{plain}
\bibliography{references}

\appendix

\clearpage
\section{Implementation Details}\label{apx:implementation}

Here we prove Lemma~\ref{lem:linear-implementation} using similar ideas to \cite[Section~4]{OnakSODA12}. Our modifications to the proof of  \cite{OnakSODA12} are straightforward and we claim no novelty in this section. As already discussed in Section~\ref{sec:adjlist}, the first (standard) idea is to generate the random permutation $\pi$ locally by sampling an independent rank $\sigma_e$ for each edge $e$ which is a real in $[0, 1]$ chosen uniformly at random, and forming $\pi$ by sorting the edges in the increasing order of their ranks.

\newcommand{\exposed}[1]{\ensuremath{\normalfont \texttt{exposed}(#1)}}
\newcommand{\nextint}[1]{\ensuremath{k(#1)}}
\newcommand{\opennext}[1]{\ensuremath{\normalfont\texttt{expose\_next}(#1)}}
\newcommand{\lowest}[1]{\ensuremath{\normalfont\texttt{lowest}(#1)}}
\newcommand{\exposedcount}[1]{\ensuremath{\normalfont\texttt{count}(#1)}}

Now suppose w.l.o.g. that $\Delta$ is a power of 2 and consider the following sub-intervals of $[0, 1)$:\footnote{If $\Delta$ is not a power of two, just take the smallest power of two that is larger than $\Delta$.}
$$
I_0 := \left[0, 1/\Delta\right), \qquad I_1 := \left[1/\Delta, 2/\Delta\right), \ldots, I_i := \left[2^{i-1}/\Delta, 2^i/\Delta\right), \ldots, I_{\log_2 \Delta} := \left[0.5, \,\, 1\right),
$$ 
and define $s_i$ such that $I_i = [s_i, s_{i+1})$. For any vertex $v$, we maintain the following data structures:
\begin{itemize}[topsep=2pt, itemsep=0pt]
	\item \exposed{v}: This is a subset of the neighbor set $N(v)$ of $v$ in graph $G$ that is initially empty. For any $u \in \exposed{v}$, vertex $v$ knows the rank $\sigma_{(v, u)}$ of the edge between them (but $v$ may not know the location of $u$ in its adjacency list as this edge may have been notified to $v$ by $u$). We store $\exposed{v}$ in two binary search trees, one indexed by the vertex IDs and the other indexed by the rank of the edges.
	\item \nextint{v}: This is an integer in $\{0, \ldots, (\log_2 \Delta) + 1\}$ which is initially 0 for every vertex $v \in V$. For any $v \in V$ and any edge  $e=(v, u)$ of $v$ with rank $\sigma_e \in I_0 \cup \ldots \cup I_{k(v)-1}$, it will hold that $u \in \exposed{v}$ and so $v$ is ``aware'' of this edge $e$ and its rank. Since initially $\nextint{v} = 0$ and $s_0 = 0$ for all $v \in V$, the vertices do not need to be aware of any of their edges initially.
	\item \exposedcount{v, i}: The number of neighbors $u$ of $v$ in $\exposed{v}$ such that $\sigma_{(v, u)} \in I_i$.
\end{itemize}

Armed with these data structures, we will describe a procedure $\lowest{v, i}$ which returns a vertex $w$ where $(v, w)$ is the $i$-th lowest rank edge of $v$. We first formalize how we can implement the oracles using this procedure and then formalize it.

\begin{algorithm}[H]\label{alg:imp-vertexoracle}
\caption{Implementation of the vertex oracle $\vertexoracle{v}$.}
	$j \gets 1$.
	
	$w \gets \lowest{v, j}$.
	
	\While{$j \leq \deg(v)$}{
		\lIf{$\edgeoracle{(v, w), w} = \true$}{\Return \true}
		
		$j \gets j + 1$.
		
		$w \gets \lowest{v, j}$.
	}
	\Return \false
\end{algorithm}

\smallskip

\begin{algorithm}[H]\label{alg:imp-edgeoracle}
\caption{Implementation of the edge oracle $\edgeoracle{e=(u, v), u}$.}

	\lIf{we have already computed $\edgeoracle{e, u}$}{\Return the computed answer.}\label{line:cache}
	
	$j \gets 1$.
	
	$w \gets \lowest{u, j}$.
	
	\While{$w \not= v$}{
			\lIf{$\edgeoracle{(u, w), w} = \true$}{\Return \false.}
			
			$j \gets j + 1$, 
			
			$w \gets \lowest{u, j}$.
	}
	\Return \true
\end{algorithm}

%
%
%

Next, we turn to formalize how we implement function $\lowest{v, i}$. To do this, we need a procedure $\opennext{v}$ which increases $k(v)$ by one and ensures that the invariants of the data structures continue to hold. That is, to increase $k(v)$ by one, we first make sure that $v$ is made aware of all of its edges in interval $I_{k(v)}$ by having them stored in \exposed{v}. Once we have $\opennext{v}$, procedure $\lowest{v, i}$ can be implemented as follows: 

\begin{algorithm}[H]\label{alg:lowest}
\caption{$\lowest{v, i}$}	
	\lIf{$i \geq \deg(v)$}{\Return $\emptyset$.}
	\While{$\exposedcount{v, 1} + \ldots + \exposedcount{v, k(v)-1} < i$}{
		\opennext{v}.
	}
	
	Let $(v, w)$ be the $i$-th lowest-rank edge of $v$ in $\exposed{v}$. This can be found in $O(\log \Delta)$ time using the BST that $\exposed{v}$ is stored in which is indexed by the ranks.
	 
	\Return $w$.
\end{algorithm}

\paragraph{Procedure $\opennext{v}$:} To expose all the edges of $v$ in interval $I_{k(v)}$, we first take a subsample $S(v, k(v))$ of $[\deg(v)]$ by including each element independently with prob. $p := |I_{\nextint{v}}|/(1 - s_{\nextint{v}})$. Note that $p$ is the probability that the rank of an edge is in $I_{\nextint{v}}$ conditioned on that its rank is in $I_{k(v)} \cup \ldots \cup I_{\log_2 \Delta}$. The goal is to assign a uniform rank from $I_{\nextint{v}}$ to any edge of $v$ whose index in $v$'s adjacency list is sampled in $S(v, k(v))$, but we have to be careful that these ranks do not contradict the previously drawn ranks. To do this, for each sampled index $i \in S(v, k(v))$, we query the $i$-th neighbor $u$ of $v$. If $u \in \exposed{v}$, we do not reveal a new rank for $(u, v)$ as $\sigma((u, v))$ is already drawn. Otherwise, we consider the two cases $k(u) \leq k(v)$ and $k(u) \geq k(v) + 1$ individually. In the former case, we go ahead and draw a rank $\sigma((u, v)) \in I_{\nextint{v}}$ uniformly at random, add $u$ to $\exposed{v}$, increase $\exposedcount{v, k(v)}$ by one, add $v$ to $\exposed{u}$, and increase $\exposedcount{u, k(v)}$ by one. In the latter case, however, the fact that $k(u) \geq k(v) + 1$ implies that $u$ is already aware of all of its edges with ranks in $I_{k(v)}$ which combined with $u \not\in \exposed{v}$ implies that $(u, v)$ should not have a rank in $I_{k(v)}$. Hence, in this case we do not reveal a new rank. At the end, we increase $k(v)$ by one. It can be confirmed that with this process, the rank $\sigma(e)$ of any edge $e$ in the graph is drawn independently and uniformly at random from $[0, 1]$.

To implement \opennext{v} efficiently, instead of forming the subsample $S(v, k(v))$ of $[\deg(v)]$ by going over the elements in $[\deg(v)]$ one by one and sampling each independently with probability $p$ which takes $\Theta(\deg(v))$ time, we first draw its size $|S(v, k(v))|$ from the binomial distribution $B(\deg(v), p)$. This can be done in $\polylog n$ time both in expectation and with probability $1-1/\poly(n)$, using the sublinear time {\em exact} binomial samplers of \cite{BinomialSampler,BinomialSampler2}.\footnote{Particularly, see Theorem~2 of \cite{BinomialSampler} and the paragraph next to it which refers to \cite{BinomialSampler2}.} Then, we pick an $|S(v, k(v))|$-size subset $S(v, k(v))$ of $[\deg(v)]$ uniformly at random which can be done in $\widetilde{O}(|S(v, k(v))|)$ time. 

We now turn to analyze the time-complexity of the algorithm. We start with two claims.

\begin{claim}\label{cl:exposenext-time}
	Consider a vertex $v$ and suppose that we run $\opennext{v}$ when $k(v)$ equals some number $k$; then procedure $\opennext{v}$ takes only  $\widetilde{O}(\deg_{k-1}(v) + 1)$ time where $\deg_i(v)$ is the number of edges of $v$ that end up receiving ranks in $I_i$ (in case $k=0$ define $\deg_{-1}(v) = 0$). This guarantee holds, w.h.p., throughout the algorithm for all calls to $\opennext{v}$ and all $v \in V$.
\end{claim}
\begin{proof}
	As discussed, $\opennext{v}$ runs in time $\widetilde{O}(|S(v, k)|)$. Moreover,
	$$
	\E|S(v, k)| = \frac{|I_{k}|}{1 - s_{k}} \deg(v) \leq 2|I_{k}| \deg(v),
	$$
	where the inequality follows from $s_{k(v)} \leq 1/2$ which itself follows from $k(v) \leq \log_2 \Delta$, noting that if $k(v) = \log_2 \Delta + 1$ we do not call $\opennext{v}$. On the other hand, we have
	$$
	\E[\deg_{k-1}(v)] = |I_{k-1}|\deg(v) \geq (|I_k| \deg(v) - 1) / 2.
	$$
	Since $S(v, k)$ and $\deg_k(v)$ are both sums of independent Bernoulli random variables, Chernoff bound asserts they are highly concentrated around their expectations. A union bound over all $n$ vertices and the $O(\log \Delta)$ times that we may call $\opennext{\cdot}$ for each finishes the proof.
\end{proof}

\begin{claim}\label{cl:lowest-time}
Consider a vertex $v$ and suppose that we call procedures $\lowest{v, 1}, \ldots, \lowest{v, j}$ for some $1 \leq j \leq \deg(v)$. The total time spent implementing all these $j$ calls is upper bounded by $\widetilde{O}(j)$. The guarantee holds, w.h.p., for all $v \in V$ and all $j$ throughout the algorithm. 	
\end{claim}
\begin{proof}
	Let $K$ be the final value of $k(v)$ after running $\lowest{v, 1}, \ldots, \lowest{v, j}$. The condition of the while-loop in Algorithm~\ref{alg:lowest} throughout all these $j$ calls to $\lowest{v, \cdot}$ is evaluated at most $K+j$ times since at most $K$ turn out to be true and exactly $j$ are false. Each evaluation takes $O(K)$ time. As such, the total running time spent directly in $\lowest{v, 1}, \ldots, \lowest{v, j}$ (i.e., ignoring the time spent in $\opennext{v}$) is $O((K + j)K) = \widetilde{O}(j)$ since $K \leq \log_2 \Delta + 1$.
	
	By Claim~\ref{cl:exposenext-time}, the $K$ calls to $\opennext{v}$ take $\widetilde{O}(\deg_{-1}(v) + \ldots + \deg_{K-2}(v))$ time. On the other hand, $\deg_{-1}(v) + \ldots + \deg_{K-2}(v) < j$ or otherwise, after the first $K-1$ calls to $\opennext{v}$, we get $\exposedcount{v, 0} + \ldots + \exposedcount{v, K-2} \geq j$ and so we do not call $\opennext{v}$ another time. Hence, the overall time-complexity is indeed $\widetilde{O}(j)$.
\end{proof}

Now we turn to analyze the total time-complexity of the algorithm. Consider a call to the edge oracle \edgeoracle{e=(u, v), u} of Algorithm~\ref{alg:imp-edgeoracle} and let $j$ be the final value of $j$ in Algorithm~\ref{alg:imp-edgeoracle}. Observe that the time spent directly in $\edgeoracle{e, u}$ is $O(j)$. We also call $\lowest{u, 1}, \ldots, \lowest{u, j-1}$ and generate $O(j)$ recursive calls to Algorithm~\ref{alg:imp-edgeoracle}. As we showed in Claim~\ref{cl:lowest-time}, the total time needed to execute the calls to $\lowest{u, 1}, \ldots, \lowest{u, j}$ is $\widetilde{O}(j)$. Hence, the time spent directly in $\edgeoracle{e, u}$ as well as the cost of calling $\lowest{\cdot, \cdot}$ by it is near-linear in the number of new recursive calls to Algorithm~\ref{alg:imp-edgeoracle} generated directly by \edgeoracle{e, u}. Similarly, the time spent directly in the vertex oracle $\vertexoracle{v}$ and its calls to $\lowest{\cdot, \cdot}$ is near-linear in the number of edge-oracle calls generated directly by $\vertexoracle{v}$. As such, the total time-complexity of the algorithm is near-linear in the total number of calls to the edge-oracle of Algorithm~\ref{alg:imp-edgeoracle}.

To finalize the proof, let $\pi$ be the permutation corresponding to $\sigma$ obtained by sorting the edges in the increasing order of their $\sigma$-ranks. Suppose that we run the vertex oracle $\vertexoracle{v, \sigma}$ in Algorithm~\ref{alg:imp-vertexoracle} and also run the original vertex oracle $\vertexoracle{v}$ of Algorithm~\ref{alg:vertexoracle} with respect to $\pi$. Observe that the number of recursive calls to Algorithm~\ref{alg:imp-edgeoracle} via Algorithm~\ref{alg:imp-vertexoracle} is exactly equal to the number of recursive calls to Algorithm~\ref{alg:edgeoracle} via Algorithm~\ref{alg:vertexoracle} which is $T(v, \pi)$. By the discussion of the previous paragraph, this means that the algorithm has time complexity $\widetilde{O}(T(v, \pi))$.

\end{document}